\documentclass[12pt,technote, onecolumn, draft]{IEEEtran}
\usepackage{latexsym}
\usepackage{mathrsfs}
\usepackage{multirow}
\usepackage{amsfonts,amssymb,amsmath,amsthm,bm}
\usepackage{color}
\usepackage{cite}
\usepackage{comment}
\newtheorem{theorem}{Theorem}

\newtheorem{definition}{Definition}
\newtheorem{remark}{Remark}
\newtheorem{corollary}{Corollary}
\newtheorem{lemma}{Lemma}
\newtheorem{proposition}{Proposition}

\begin{document}

\title{Constructions and Characterizations of $s$-Plateaued Partitions $^{\dag}$}
\author{Jiaxin Wang, Yadi Wei, Fang-Wei Fu, Jin Li, Fulin Li
\IEEEcompsocitemizethanks{\IEEEcompsocthanksitem Jiaxin Wang, Jin Li and Fulin Li are with the School of Mathematics, Hefei University of Technology, Hefei 230601, China, Emails: wjiaxin@hfut.edu.cn, lijin\_0102@126.com, lflsxx66@163.com, Yadi Wei and Fang-Wei Fu are with the Chern Institute of Mathematics and LPMC, Nankai University, Tianjin 300071, China, Emails: wydecho@mail.nankai.edu.cn, fwfu@nankai.edu.cn.
}
\thanks{$^\dag$This research is supported by the National Key Research and Development Program of China (Grant No. 2022YFA1005000), the National Natural Science Foundation of China (Grant Nos. 62371259 and 12471491), the Fundamental Research Funds for the Central Universities of China (Nankai University), and the Nankai Zhide Foundation.}
\thanks{manuscript submitted June 26, 2026}
}

\maketitle

\begin{abstract}
 Bent partitions play a significant role in constructing bent functions and have rich connections with coding theory and combinatorics. In this paper, we introduce $s$-plateaued partitions, which generalize the bent partitions. Let $\Gamma=\{A_{i}, 1 \leq i \leq K\}$ be a partition of $V_{n}^{(p)}$, where $V_{n}^{(p)}$ is an $n$-dimensional vector space over the prime field $\mathbb{F}_{p}$ and $p \mid K$. Then $\Gamma$ is called an $s$-plateaued partition of $V_{n}^{(p)}$ of depth $K$ if each $p$-ary function $f: V_{n}^{(p)} \rightarrow \mathbb{F}_{p}$ for which every $j \in \mathbb{F}_{p}$ has exactly $\frac{K}{p}$ of sets $A_{i}$ in $\Gamma$ in its preimage set, is a $p$-ary $s$-plateaued function. By using an $s$-plateaued partition, a large number of $p$-ary $s$-plateaued functions, vectorial $s$-plateaued functions and generalized $s$-plateaued functions can be constructed. In particular, $0$-plateaued partitions are just bent partitions. In general, $s$-plateaued partitions are much more complicated than bent partitions. We analyze the possible cardinality of $A_{i}$ of an $s$-plateaued partition. We give some explicit constructions of $s$-plateaued partitions for which any generated $p$-ary $s$-plateaued function has no nonzero linear structure. We give a characterization of an $s$-plateaued partition $\Gamma=\{A_{i}, 1 \leq i \leq K\}$, where $p$ is odd, $K \geq 5$ and $-A_{i}=A_{i}, 1 \leq i \leq K$. Based on which, we show that if $p \geq 5$, then the preimage set partition of a $p$-ary $s$-plateaued function $f: V_{n}^{(p)} \rightarrow \mathbb{F}_{p}$ with $f(x)=f(-x)$ is an $s$-plateaued partition if and only if $f$ is of $(p-1)$-form, where $n+s$ is even.
 When $s=0$, we partially address an open problem on whether a bent partition $\Gamma$ of $V_{n}^{(p)}$ of depth $p^{\frac{n}{2}}$ must be obtained from spreads. 
\end{abstract}


\section{Introduction}

As an important class of cryptographic functions,  $s$-plateaued functions were introduced in \cite{ZZ2001On}. 
In particular, $0$-plateaued functions are bent functions. Plateaued functions can possess many good cryptographic properties, such as high nonlinearity, resiliency, low autocorrelation, without nonzero linear structure and satisfying propagation criteria (see e.g. \cite{Mesnager2016Be}). Besides their desirable cryptographic properties, plateaued functions have important applications in coding theory, sequences and combinatorics (see e.g. \cite{BFT2018Ex,MOS2019Li,MS2020Se,Olmez2015Pl,WWF2026Se}). As generalizations, vectorial $s$-plateaued functions \cite{Carlet2015Bo,MOS2016Re} and generalized $s$-plateaued functions \cite{MTQ2017Ge} were introduced. Vectorial $s$-plateaued functions can be used to construct partial geometric difference sets \cite{CO2021Gr}, phase orthogonal sequence sets for (QS)CDMA communications \cite{ZPZ2022Ph} and  self-orthogonal minimal codes \cite{RPZW2024Se}. 
 
 Let $V_{n}^{(p)}$ be an $n$-dimensional vector space over the prime field $\mathbb{F}_{p}$ and $K$ be a positive integer divisible by $p$. A partition $\Gamma=\{A_{i}, 1 \leq i \leq K\}$ of $V_{n}^{(p)}$ is said to be a bent partition if any $p$-ary function $f: V_{n}^{(p)} \rightarrow \mathbb{F}_{p}$ that takes each value $j \in \mathbb{F}_{p}$ on exactly of $\frac{K}{p}$ of sets $A_{i}$ is a $p$-ary bent function \cite{AM2022Be}. Bent partitions play a quite powerful role in constructing $p$-ary bent functions, vectorial $p$-ary bent functions and generalized $p$-ary bent functions. Bent partitions also have important applications in constructing linear codes, partial difference sets, LP-packings,
 association schemes and Hadamard matrices (see e.g. \cite{AAKM2024Be,AKM2024Am,AKM2022Be,AKMO2025Be,WFWY2024A}).

In this paper, we generalize bent partitions to $s$-plateaued partitions. For a partition  $\Gamma=\{A_{i}, 1 \leq i \leq K\}$ of $V_{n}^{(p)}$, if each $p$-ary function $f: V_{n}^{(p)} \rightarrow \mathbb{F}_{p}$ that takes every value $j \in \mathbb{F}_{p}$ on exactly of $\frac{K}{p}$ of sets $A_{i}$ is a $p$-ary $s$-plateaued function, then $\Gamma$ is called an $s$-plateaued partition. A $0$-plateaued partition is just a bent partition. Utilizing a non-trivial $s$-plateaued partition allows us to generate numerous $p$-ary $s$-plateaued functions, vectorial $p$-ary $s$-plateaued functions and generalized $p$-ary $s$-plateaued functions. Any $p$-ary bent function from $V_{n}^{(p)}$ to $\mathbb{F}_{p}$ is unbalanced and its Walsh support is $V_{n}^{(p)}$. Overall, due to the fact that plateaued functions can be either balanced or unbalanced, and the Walsh support of a plateaued function is in general indeterminable, $s$-plateaued partitions are far more intricate than bent partitions. We present the possible cardinality of $A_{i}$ of an $s$-plateaued partition. We provide a general way to construct $s$-plateaued partitions. Based on which, some explicit constructions of $s$-plateaued partitions are given. Particularly, some constructions produce a large number of $p$-ary $s$-plateaued functions without nonzero linear structure. When $p$ is odd, $K \geq 5$, we provide a characterization of $s$-plateaued partitions $\Gamma=\{A_{i}, 1 \leq i \leq K\}$ with  $-A_{i}=A_{i}, 1 \leq i \leq K$. For a $p$-ary $s$-plateaued function $f: V_{n}^{(p)} \rightarrow \mathbb{F}_{p}$ with $f(x)=f(-x)$, where $p \geq 5$, we show that the preimage set partition of $f$ is an $s$-plateaued partition if and only if $f$ is of $(p-1)$-form, where $n+s$ is even. When $s=0$, we partially address an open problem on the relationship between bent partitions of depth $p^{\frac{n}{2}}$ and spreads.

The rest of the paper is organized as follows. Section II provides the needed preliminaries. In Section III, we introduce the definition of $s$-plateaued partitions. 
We analyze the possible cardinality of $A_{i}$ of an $s$-plateaued partition $\Gamma=\{A_{i}, 1 \leq i \leq K\}$. In Section IV, a general way to construct $s$-plateaued partitions is provided. Then we present some explicit constructions of $s$-plateaued partitions for which any generated $p$-ary $s$-plateaued function has no nonzero linear structure. In Section V, we give a characterization of $s$-plateaued partitions $\Gamma=\{A_{i}, 1 \leq i \leq K\}$ with $-A_{i}=A_{i}, 1 \leq i \leq K$, where  $p$ is odd and $K \geq 5$. We investigate the relations between $s$-plateaued partitions of depth $p$ with $-A_{i}=A_{i}$ and $p$-ary $s$-plateaued functions of $(p-1)$-form, where $p \geq 5$. Furthermore, we partially resolve an open problem regarding the relations between bent partitions of depth $p^{\frac{n}{2}}$ and spreads. In Section VI, we conclude this paper and list some problems for further research.

\section{Preliminaries}
In this section, we present some results on plateaued functions, number fields and presemifields. First, we fix some notations that are used throughout this paper.
\subsection{Some notations}
\begin{itemize}
	\item Let $p$ be a prime.
	\item For an integer $m$, let $\zeta_{m}$ be a complex primitive $m$-th root of unity.
	\item Let $\mathbb{F}_{p^n}$ be the finite field with $p^n$ elements.
	\item Let $\mathbb{F}_{p}^{n}$ be the vector space of $n$-tuples over $\mathbb{F}_{p}$.
	\item Let $V_{n}^{(p)}$ be an $n$-dimensional vector space over $\mathbb{F}_{p}$. 
	\item Let $\langle,\rangle_{n}$ be an inner product of $V_{n}^{(p)}$. If $V_{n}^{(p)}=\mathbb{F}_{p}^{n}$, let $\langle x, y\rangle_{n}=x \cdot y$, where $\cdot$ is the usual dot product. If $V_{n}^{(p)}=\mathbb{F}_{p^n}$, let $\langle x, y\rangle_{n}=\mathrm{Tr}_{1}^{n}(xy)$, where $\mathrm{Tr}_{m}^{n}$ is the trace function from $\mathbb{F}_{p^n}$ to $\mathbb{F}_{p^m}$ ($m$ is a divisor of $n$). If $V_{n}^{(p)}=V_{n_{1}}^{(p)} \times \cdots \times V_{n_{t}}^{(p)}$, let $\langle x, y\rangle_{n}=\sum_{i=1}^{t}\langle x_{i}, y_{i}\rangle_{n_{i}}$, where $x=(x_{1}, \dots, x_{t}), x_{i} \in V_{n_{i}}^{(p)}, 1 \leq i \leq t$.
	\item Let $\mathbb{Z}_{m}$ be the ring of integers modulo $m$.
	\item For $F: A \rightarrow B$ and $I \subseteq B$, let $D_{F, I}=\{x \in A: F(x) \in I\}$. When $I=\{i\}$, simply denote $D_{F, \{i\}}$ by $D_{F, i}$. If $F$ is injective, let $F^{-1}(b)=a$, where $F(a)=b$.
	\item For $S \subseteq V_{n}^{(p)}$, let $S^{*}=\{x \in S: x \neq 0\}$, and let $\chi_{a}(S)=\sum_{x \in S}\chi_{a}(x)$, where $\chi_{a}(x)=\zeta_{p}^{\langle a, x\rangle_{n}}$, $a \in V_{n}^{(p)}$.
	\item Let $\delta_{S}(x)$ be the indicator function. If $S=\{i\}$, simply denote $\delta_{\{i\}}(x)$ by $\delta_{i}(x)$. 
	\item For $x=0 \in \mathbb{F}_{p^n}$, for convention we denote $x^{-1}=x^{p^n-2}=0$.
	\item Let $\eta_{n}$ be the quadratic character of $\mathbb{F}_{p^n}$, that is, $\eta_{n}(a)=1$ if $a$ is a square element in $\mathbb{F}_{p^n}^{*}$, and $\eta_{n}(a)=-1$ if $a$ is a non-square element in $\mathbb{F}_{p^n}^{*}$. For convention, we denote $\eta_{n}(0)=0$.
\end{itemize}

\subsection{Some results on plateaued functions}

Let $f$ be a $p$-ary function from $V_{n}^{(p)}$ to $\mathbb{F}_{p}$. If there exists an integer $1 \leq \ell \leq p-1$ such that $f(ax)=a^{\ell}f(x)$ for any $a \in \mathbb{F}_{p}^{*}$ and $x \in V_{n}^{(p)}$, then $f$ is called of $\ell$-form. The \emph{Walsh transform} $W_{f}$ of $f: V_{n}^{(p)}\rightarrow \mathbb{F}_{p}$ is defined by
\begin{equation*}
	W_{f}(a)=\sum_{x \in V_{n}^{(p)}}\zeta_{p}^{f(x)-\langle a, x\rangle_{n}} \text{ for any } a \in V_{n}^{(p)}.
\end{equation*}
The $p$-ary function $f$ can be recovered by the inverse Walsh transform
\begin{equation*}
	\zeta_{p}^{f(x)}=\frac{1}{p^n}\sum_{a \in V_{n}^{(p)}}W_{f}(a)\zeta_{p}^{\langle a, x\rangle_{n}} \text{ for any } x \in V_{n}^{(p)}.
\end{equation*}
Denote $S_{f}=\{a \in V_{n}^{(p)}: W_{f}(a) \neq 0\}$, which is said to be the \emph{Walsh support} of $f$. Notice that $f$ is balanced (that is, $|D_{f, i}|=p^{n-1}, i \in \mathbb{F}_{p}$) if and only if $0 \notin S_{f}$. If $|W_{f}(a)|=p^{\frac{n+s}{2}}$ for all $a \in S_{f}$, then $f$ is called a \emph{$p$-ary $s$-plateaued} function. When $p=2$, that is, $f$ is a Boolean $s$-plateaued function, then $n+s$ must be even. For a $p$-ary $s$-plateaued function $f: V_{n}^{(p)}\rightarrow \mathbb{F}_{p}$, $|S_{f}|=p^{n-s}$ by the Parseval's identity $\sum_{a \in V_{n}^{(p)}}|W_{f}(a)|^{2}=p^{2n}$. A $0$-plateaued function $f: V_{n}^{(p)}\rightarrow \mathbb{F}_{p}$ is called a \emph{bent} function and $S_{f}=V_{n}^{(p)}$. An $n$-plateaued function is an affine function. In this paper, we only consider $s$-plateaued functions, where $s \neq n$. For a $p$-ary $s$-plateaued function $f: V_{n}^{(p)} \rightarrow \mathbb{F}_{p}$, when $p=2$, then
\begin{equation*}
	W_{f}(a)=2^{\frac{n+s}{2}}(-1)^{f^{*}(a)} \text{ for any } a \in S_{f},
\end{equation*}
and when $p$ is an odd prime, then for $a \in S_{f}$,
\begin{equation*}
	W_{f}(a)=\left\{\begin{split}
		\pm p^{\frac{n+s}{2}}\zeta_{p}^{f^{*}(a)}, & \text{ if } p^{n+s} \equiv 1 \pmod 4,\\
		\pm \sqrt{-1} p^{\frac{n+s}{2}} \zeta_{p}^{f^{*}(a)}, & \text{ if } \ p^{n+s} \equiv 3 \pmod 4,
	\end{split}\right.
\end{equation*}
where $f^{*}$ is a function from $S_{f}$ to $\mathbb{F}_{p}$, called the \textit{dual} of $f$. For a $p$-ary $s$-plateaued function $f: V_{n}^{(p)}\rightarrow \mathbb{F}_{p}$, let
\begin{equation*}
	\epsilon_{f}(x)=p^{-\frac{n+s}{2}}\zeta_{p}^{-f^{*}(x)}W_{f}(x), x \in S_{f}.
\end{equation*}
Note that $\epsilon_{f}(x) \in \{\pm 1, \pm \sqrt{-1}\}$. If $\epsilon_{f}$ is a constant function, then $f$ is called \textit{weakly regular}, otherwise $f$ is called \textit{non-weakly regular}. In particular, if $\epsilon_{f}=1$, then $f$ is called \textit{regular}. All Boolean $s$-plateaued functions are regular.

A vectorial $p$-ary function $F: V_{n}^{(p)}\rightarrow V_{m}^{(p)}$ is called \textit{vectorial $s$-plateaued} if all \textit{component functions} $F_{c}: V_{n}^{(p)}\rightarrow \mathbb{F}_{p}, c \in {V_{m}^{(p)}}^{*}$, defined as $F_{c}(x)=\langle c, F(x)\rangle_{m}$ are $p$-ary $s$-plateaued functions. Every $p$-ary $s$-plateaued function is vectorial $s$-plateaued. A vectorial $0$-plateaued function is called \emph{vectorial bent}. For a vectorial bent function $F: V_{n}^{(p)}\rightarrow V_{m}^{(p)}$,  if there exists a vectorial bent function $G: V_{n}^{(p)}\rightarrow V_{m}^{(p)}$ such that $(F_{c})^{*}=G_{\sigma(c)}, c \in {V_{m}^{(p)}}^{*}$, where $\sigma$ is some permutation of $V_{m}^{(p)}$ with $\sigma(0)=0$, then $F$ is called \emph{vectorial dual-bent}. The vectorial bent function $G$ is called a \textit{vectorial dual} of $F$ and denoted by $F^{*}$. We call a vectorial dual-bent function $F: V_{n}^{(p)} \rightarrow V_{m}^{(p)}$ satisfying \emph{Condition A} if there exists a vectorial dual $F^{*}$ such that $(F_{c})^{*}=(F^{*})_{c}, c \in {V_{m}^{(p)}}^{*}$, and $F_{c}$ are all weakly regular with $\epsilon_{F_{c}}=\epsilon$, where $\epsilon \in \{\pm 1\}$ is a constant. For the constructions of vectorial dual-bent functions satisfying Condition A, please refer to \cite{AAKM2025Be, AFKMWW2026A,WFW2023Be, WWF2026Fu}.

For a generalized $p$-ary function $f: V_{n}^{(p)} \rightarrow \mathbb{Z}_{p^m}$, define the generalized Walsh transform $W_{f}$ as $W_{f}(a)=\sum_{x \in V_{n}^{(p)}}\zeta_{p^m}^{f(x)}\zeta_{p}^{-\langle a, x\rangle_{n}}, a \in V_{n}^{(p)}$. If $|W_{f}(a)| \in \{p^{\frac{n+s}{2}}, 0\}$ for any $a \in V_{n}^{(p)}$, then $f$ is called a \emph{generalized $p$-ary $s$-plateaued} function. If $m=1$, then generalized $p$-ary $s$-plateaued functions reduce to $p$-ary $s$-plateaued functions.

The following lemma gives the value distribution of unbalanced
$p$-ary $s$-plateaued functions. 

\begin{lemma}[\cite{OP2020Du}]\label{pre1}
Let $f: V_{n}^{(p)} \rightarrow \mathbb{F}_{p}$ be a unbalanced $p$-ary $s$-plateaued function, that is, $0 \in S_{f}$.
\begin{itemize}
	\item When $p=2$, then $$|D_{f, 0}|=2^{n-1}+(-1)^{f^{*}(0)}2^{\frac{n+s}{2}-1}, |D_{f, 1}|=2^{n-1}-(-1)^{f^{*}(0)}2^{\frac{n+s}{2}-1}.$$
	\item When $p$ is odd and $n+s$ is even, then $$|D_{f, f^{*}(0)}|=p^{n-1}-\epsilon_{f}(0)p^{\frac{n+s}{2}-1}+\epsilon_{f}(0)p^{\frac{n+s}{2}}, |D_{f, i}|=p^{n-1}-\epsilon_{f}(0)p^{\frac{n+s}{2}-1}, i \in \mathbb{F}_{p} \backslash \{f^{*}(0)\}.$$
	\item When $p$ is odd and $n+s$ is odd, then $$|D_{f, f^{*}(0)}|=p^{n-1}, |D_{f, i}|=p^{n-1}+\epsilon_{f}(0)\sqrt{\eta_{1}(-1)}p^{\frac{n+s-1}{2}}\eta_{1}(f^{*}(0)-i), i \in \mathbb{F}_{p} \backslash \{f^{*}(0)\}.$$
\end{itemize}
\end{lemma}

For a $p$-ary function $f: V_{n}^{(p)}\rightarrow \mathbb{F}_{p}$, if $f(x+a)-f(x)$ is a constant function, then $a$ is called a \emph{linear structure} of $f$. Any $p$-ary bent function has no nonzero linear structure. The following proposition generalizes \cite[Corollary 3.1]{HPWZ2019De} of the case of $p=2$ to any characteristic.

\begin{proposition}\label{Pronew}
	Let $f: V_{n}^{(p)}\rightarrow \mathbb{F}_{p}$ be an $s$-plateaued function and $S_{f}$ be its Walsh support. Denote $S_{f}=v+E$, where $v \in S_{f}$. Then $f$ has no nonzero linear structure if and only if $E$ contains a basis of $V_{n}^{(p)}$.
\end{proposition}

\begin{proof}
	For any $a \in V_{n}^{(p)}$,
	\begin{equation}\label{1}
		\begin{split}
			\sum_{x \in V_{n}^{(p)}}\zeta_{p}^{f(x+a)-f(x)}&=\frac{1}{p^n}\sum_{x, y \in V_{n}^{(p)}}\zeta_{p}^{f(y)-f(x)}\sum_{z \in V_{n}^{(p)}}\zeta_{p}^{\langle z, x-y+a\rangle_{n}}\\
			&=\frac{1}{p^n}\sum_{z \in V_{n}^{(p)}}\zeta_{p}^{\langle z, a\rangle_{n}}\sum_{x \in V_{n}^{(p)}}\zeta_{p}^{-f(x)+\langle z, x\rangle_{n}}\sum_{y \in V_{n}^{(p)}}\zeta_{p}^{f(y)-\langle z, y\rangle_{n}}\\
		&=\frac{1}{p^n}\sum_{z \in V_{n}^{(p)}}|W_{f}(z)|^{2}\zeta_{p}^{\langle z, a\rangle_{n}}=\frac{1}{p^n}\sum_{z \in S_{f}}p^{n+s}\zeta_{p}^{\langle z, a\rangle_{n}}=p^{s}\zeta_{p}^{\langle v, a\rangle_{n}}\sum_{z \in E}\zeta_{p}^{\langle z, a\rangle_{n}}.
		\end{split}
	\end{equation}
	It is easy to see that $a$ is a linear structure if and only if $|\sum_{x \in V_{n}^{(p)}}\zeta_{p}^{f(x+a)-f(x)}|=p^n$. Then by Eq. \eqref{1}, $a$ is a linear structure if and only if $|\sum_{z \in E}\zeta_{p}^{\langle z, a\rangle_{n}}|=p^{n-s}$. Note that $|E|=p^{n-s}$ and $0 \in E$. Thus, $a$ is a linear structure if and only if $\langle z, a\rangle_{n}=0$ for all $z \in E$, if and only if $a \in E^{\bot}$, where $E^{\bot}=\{x \in V_{n}^{(p)}: \langle z, x\rangle_{n}=0 \text{ for all } z \in E\}$ is a linear subspace. As easily seen, $E^{\bot}=<E>^{\bot}$, where $<E>$ is the linear space spanned by $E$. Then $E^{\bot}=\{0\}$ if and only if $E$ contains a basis of $V_{n}^{(p)}$.
\end{proof}

\subsection{Some results on number fields}
In this subsection, we present some results on number fields. For more information on number fields, the reader is referred to \cite{Feng2018Al,IR1990A}.
 
Let $L$ be a number field, that is, $L$ is a finite dimensional extension field of the rational number
field $\mathbb{Q}$. For $\beta \in L$, if there exists a monic polynomial $f(x) \in \mathbb{Z}[x]$ such that $f(\beta)=0$, then $\beta$ is called an algebraic integer. The set of all algebraic integers in $L$ forms a ring, denoted by $\mathcal{O}_{L}$. It is clear that $\mathcal{O}_{\mathbb{Q}}=\mathbb{Z}$. Each nonzero ideal $I$ of $\mathcal{O}_{L}$ with $I \neq \mathcal{O}_{L}$ can be uniquely (up to the order) expressed as $I=\mathcal{P}_{1}^{e_{1}}\cdots\mathcal{P}_{g}^{e_{g}}$, where $\mathcal{P}_{i}, 1 \leq i \leq g$, are distinct prime ideals of $\mathcal{O}_{L}$ and $e_{i}, 1 \leq i \leq g$, are positive integers. If $I=p\mathcal{O}_{L}$, then $p\mathbb{Z}=\mathcal{P}_{i} \cap \mathbb{Z}$ for any $1 \leq i \leq g$.

We will consider the cyclotomic field $L=\mathbb{Q}(\zeta_{p})$, where $p$ is odd. The ring of algebraic integers in $L$ is $\mathcal{O}_{L}=\mathbb{Z}[\zeta_{p}]$. Denote $p^{*}=\eta_{1}(-1)p$. Then $\sqrt{p^{*}} \in \mathcal{O}_{L}$. The group of roots of unity in $\mathcal{O}_{L}$ is $\{\pm \zeta_{p}^{i}, 0 \leq i \leq p-1\}$.
It is known that $p\mathcal{O}_{L}=((1-\zeta_{p})\mathcal{O}_{L})^{p-1}$, where $(1-\zeta_{p})$ is a prime ideal of $\mathcal{O}_{L}$. For any $1 \leq i \leq p-1$, $1+\zeta_{p}^{i}$ is a unit of $\mathcal{O}_{L}$ and $(1-\zeta_{p}^{i})\mathcal{O}_{L}=(1-\zeta_{p})\mathcal{O}_{L}$.

\subsection{Some results on presemifields}

Suppose $\circ$ is a binary operation defined on $(\mathbb{F}_{p^n}, +)$ such that $x \circ y=0$ implies $x=0$ or $y=0$, and $(x+y)\circ z=x\circ z+y \circ z, z\circ (x+y)=z\circ x+z\circ y$ for all $x, y, z \in \mathbb{F}_{p^n}$. Then $(\mathbb{F}_{p^n}, +, \circ)$ is called a \textit{presemifield}. For a presemifield $P=(\mathbb{F}_{p^n}, +, \circ)$, the presemifield $P^{T}=(\mathbb{F}_{p^n}, +, \star)$, called the \textit{transpose} of $P$, is obtained by $\mathrm{Tr}_{1}^{n}(z(x\circ y))=\mathrm{Tr}_{1}^{n}(x(z\star y)) \text{ for all }x, y, z \in \mathbb{F}_{p^n}$. If $x\circ (cy)=c(x\circ y)$ for any $x, y \in \mathbb{F}_{p^n}, c \in \mathbb{F}_{p^m}$, where $m \mid n$, then $P$ is called \textit{right $\mathbb{F}_{p^m}$-linear}. Some right $\mathbb{F}_{p^m}$-linear presemifields are presented in \cite{AKM2023Ge}.

\section{Basic properties of $s$-plateaued partitions}

We generalize the concept of bent partitions to $s$-plateaued partitions. A $0$-plateaued partition is just a bent partition introduced in \cite{AM2022Be}.

\begin{definition}\label{Definition 1}
Let $K$ be a positive integer divisible by $p$. Let $\Gamma=\{A_{i}, 1 \leq i \leq K\}$ be a partition of $V_{n}^{(p)}$. Assume that every $p$-ary function $f: V_{n}^{(p)} \rightarrow \mathbb{F}_{p}$ for which every $r \in \mathbb{F}_{p}$ has exactly $\frac{K}{p}$ of sets $A_{i}$ in $\Gamma$ in its preimage set, is a $p$-ary $s$-plateaued function. Then $\Gamma$ is called an $s$-plateaued partition of $V_{n}^{(p)}$ of depth $K$.
\end{definition}

\begin{remark}
Obviously, the preimage set partition of any Boolean $s$-plateaued function forms an $s$-plateaued partition of depth $2$. Notice that any permutation of $\mathbb{F}_{3}$ is affine. Therefore, the preimage set partition of any ternary $s$-plateaued function forms an $s$-plateaued partition of depth $3$. These $s$-plateaued partitions are viewed as trivial. Furthermore, it is easy to check that if $\Theta=\{A_{i}, 1 \leq i \leq K\}$ is a bent partition of $V_{n}^{(p)}$, then $\Gamma=\{A_{i} \times V_{s}^{(p)}, 1 \leq i \leq K\}$ is an $s$-plateaued partition of $V_{n}^{(p)} \times V_{s}^{(p)}$ for which the generated $p$-ary $s$-plateaued functions are of the form $f(x, y)=h(x)$, where $h$ is any generated $p$-ary bent function by $\Theta$. Such an $s$-plateaued partition is also viewed as trivial since it only generates partially bent functions \cite{Carlet1993Pa}. In this paper, we are interested in non-trivial $s$-plateaued partitions.
\end{remark}

By using a non-trivial $s$-plateaued partition, one can obtain a large number of vectorial $s$-plateaued functions and generalized $s$-plateaued functions.

\begin{theorem}\label{Theorem 2.1}
Let $K=p^{m} k$, where $m \geq 1$ and $p \nmid k$. Let $\Gamma=\{A_{i}, 1 \leq i \leq K\}$ be a partition of $V_{n}^{(p)}$. Then the following statements are pairwise equivalent.

(i) $\Gamma=\{A_{i}, 1 \leq i \leq K\}$ is an $s$-plateaued partition.

(ii) Let $F: V_{n}^{(p)} \rightarrow V_{m}^{(p)}$ for which every $r \in V_{m}^{(p)}$ has exactly $k$ of sets $A_{i}$ in $\Gamma$ in its preimage set. Then $F$ is a vectorial $s$-plateaued function.

(iii) Let $g: V_{n}^{(p)} \rightarrow \mathbb{Z}_{p^m}$ for which every $r \in \mathbb{Z}_{p^m}$ has exactly $k$ of sets $A_{i}$ in $\Gamma$ in its preimage set. Then $g$ is a generalized $s$-plateaued function.
\end{theorem}

\begin{proof}
If $m=1$, then the result follows from the definition of $s$-plateaued partitions. In the following, we consider the case of $m \geq 2$.

$(i) \Rightarrow (ii)$: Let $F: V_{n}^{(p)} \rightarrow V_{m}^{(p)}$ for which every $r \in V_{m}^{(p)}$ has exactly $k$ of sets $A_{i}$ in $\Gamma$ in its preimage set. Denote $D_{F, r}=\bigcup_{i=1}^{k}A_{b_{i}^{(r)}}, r \in V_{m}^{(p)}$. For any $c \in {V_{m}^{(p)}}^{*}$, let $l^{(c)}(x)=\langle c, x\rangle_{m}, x \in V_{m}^{(p)}$. Denote $D_{l^{(c)}, j}=\{x_{1}^{(c, j)}, \dots, x_{p^{m-1}}^{(c, j)}\}, j \in \mathbb{F}_{p}$. Then for any $c \in {V_{m}^{(p)}}^{*}$ and $j \in \mathbb{F}_{p}$, 
$D_{F_{c}, j}=\bigcup_{t=1}^{p^{m-1}}\bigcup_{i=1}^{k}A_{b_{i}^{(x_{t}^{(c, j)})}}$. By the definition of $s$-plateaued partitions, $F_{c}$ is a $p$-ary $s$-plateaued function, and then $F$ is a vectorial $s$-plateaued function.

$(ii) \Rightarrow (i)$: Let $f: V_{n}^{(p)} \rightarrow \mathbb{F}_{p}$ for which every $j \in \mathbb{F}_{p}$ has exactly $\frac{K}{p}$ of sets $A_{i}$ in $\Gamma$ in its preimage set. Denote $D_{f, j}=\bigcup_{i=1}^{\frac{K}{p}}A_{d_{i}^{(j)}}, j \in \mathbb{F}_{p}$. Let $c \in {V_{m}^{(p)}}^{*}$ be fixed. Let $F: V_{n}^{(p)} \rightarrow V_{m}^{(p)}$ be given by $D_{F, x_{t}^{(c, j)}}=\bigcup_{i=1+(t-1)k}^{tk}A_{d_{i}^{(j)}}, 1 \leq t \leq p^{m-1}, j \in \mathbb{F}_{p}$, where $\langle c, x_{t}^{(c, j)}\rangle_{m}=j, 1 \leq t \leq p^{m-1}, j \in \mathbb{F}_{p}$. Then every $r \in V_{m}^{(p)}$ has exactly $k$ of sets $A_{i}$ in $\Gamma$ in the preimage set of $F$, and $F$ is a vectorial $s$-plateaued function. Notice that $f(x)=F_{c}(x)$. Thus, $f$ is a $p$-ary $s$-plateaued function, and $\Gamma$ is an $s$-plateaued partition.

$(i) \Rightarrow (iii)$: Let $g: V_{n}^{(p)} \rightarrow \mathbb{Z}_{p^m}$ for which every $r \in \mathbb{Z}_{p^m}$ has exactly $k$ of sets $A_{i}$ in $\Gamma$ in its preimage set. Denote $D_{g, r}=\bigcup_{i=1}^{k}A_{b_{i}^{(r)}}, r \in \mathbb{Z}_{p^m}$. By \cite[Corollary 16]{MTQ2017Ge}, $g$ is a generalized $s$-plateaued function if and only if for any $H: \mathbb{F}_{p}^{m-1} \rightarrow \mathbb{F}_{p}$, $h(x)=g^{(0)}(x)+H(g^{(1)}(x), \dots, g^{(m-1)}(x))$ is $s$-plateaued, where $g(x)=g^{(0)}(x)p^{m-1}+g^{(1)}(x)p^{m-2}+\dots+g^{(m-2)}(x)p+g^{(m-1)}(x)$ and $g^{(t)}: V_{n}^{(p)} \rightarrow \mathbb{F}_{p}, t=0, 1, \dots, m-1$. For $j \in \mathbb{F}_{p}$, $D_{h, j}=\bigcup_{j_{1}, \dots, j_{m-1} \in \mathbb{F}_{p}}\{x \in V_{n}^{(p)}: g^{(0)}(x)=j-H(j_{1}, \dots, j_{m-1}), g^{(1)}(x)=j_{1}, \dots, g^{(m-1)}(x)=j_{m-1}\}=\bigcup_{r \in R_{j}}D_{g, r}=\bigcup_{r \in R_{j}}\bigcup_{i=1}^{k}A_{b_{i}^{(r)}}$, where $R_{j}=\{(j-H(j_{1}, \dots, j_{m-1}))p^{m-1}+j_{1}p^{m-2}+\dots+j_{m-1}: j_{1}, \dots, j_{m-1} \in \mathbb{F}_{p}\}$. Note that $|R_{j}|=p^{m-1}$. Hence, $h$ is a $p$-ary $s$-plateaued function generated by $\Gamma$ and $g$ is a generalized $s$-plateaued function.

$(iii) \Rightarrow (i)$: Let $f: V_{n}^{(p)} \rightarrow \mathbb{F}_{p}$ for which every $j \in \mathbb{F}_{p}$ has exactly $\frac{K}{p}$ of sets $A_{i}$ in $\Gamma$ in its preimage set. Denote $D_{f, j}=\bigcup_{i=1}^{\frac{K}{p}}A_{d_{i}^{(j)}}, j \in \mathbb{F}_{p}$. Let $c=(1, c_{1}, \dots, c_{m-1}) \in {\mathbb{F}_{p}^{m}}^{*}$ be fixed. Let $F=(f^{(0)}, \dots, f^{(m-1)}): V_{n}^{(p)} \rightarrow \mathbb{F}_{p}^{m}$ be given by $D_{F, x_{t}^{(c, j)}}=\bigcup_{i=1+(t-1)k}^{tk}A_{d_{i}^{(j)}}, 1 \leq t \leq p^{m-1}, j \in \mathbb{F}_{p}$, where $c \cdot x_{t}^{(c, j)}=j, 1 \leq t \leq p^{m-1}, j \in \mathbb{F}_{p}$. Then $f(x)=F_{c}(x)=f^{(0)}(x)+c_{1}f^{(1)}(x)+\dots+c_{m-1}f^{(m-1)}(x)$. Let $h: V_{n}^{(p)} \rightarrow \mathbb{Z}_{p^m}$ be defined by $h(x)=f^{(0)}(x)p^{m-1}+\dots+f^{(m-2)}(x)p+f^{(m-1)}(x)$. For any $r=r_{0}p^{m-1}+\dots+r_{m-2}p+r_{m-1} \in \mathbb{Z}_{p^m}$, where $r_{0}, \dots, r_{m-2}, r_{m-1} \in \mathbb{F}_{p}$, $D_{h, r}=D_{F, (r_{0}, \dots, r_{m-1})}, r \in \mathbb{Z}_{p^m}$. Then every $r \in \mathbb{Z}_{p^m}$ has exactly $k$ of sets $A_{i}$ in $\Gamma$ in the preimage set of $h$, and $h$ is a generalized $s$-plateaued function. Therefore, $f$ is a $p$-ary $s$-plateaued function by \cite[Corollary 16]{MTQ2017Ge}, and $\Gamma$ is an $s$-plateaued partition.
\end{proof}

\begin{remark}
When $s=0$, \cite[Theorem 5, Remark 5]{AKM2022Be} shows $(i) \Rightarrow (ii)$, \cite[Proposition 1, Remark 5]{AKM2022Be} shows $(i) \Rightarrow (iii)$, \cite[Theorem 1]{WWF2026Fu} and \cite[Remark 5]{AKM2022Be}
show $(i) \Leftrightarrow (ii)$. 
\end{remark}

The following proposition analyzes the possible cardinality of $A_{i}$ of an $s$-plateaued partition  $\Gamma=\{A_{i}, 1 \leq i \leq K\}$, which is much more complicated than bent partitions.

\begin{proposition} \label{Proposition 1}
Let $\Gamma=\{A_{i}, 1 \leq i \leq K\}$ be an $s$-plateaued partition of $V_{n}^{(p)}$, where $K \geq 4$. Without loss of generality, suppose $|A_{1}| \leq |A_{2}| \leq \cdots \leq |A_{K}|$.

(i) If $\Gamma$ generates both balanced and unbalanced $p$-ary $s$-plateaued functions simultaneously, then $p=2$, or $p=3$ and $n+s$ is odd. Furthermore, when $p=2$, the cardinality of $A_{i}$ satisfies one of the following:
\begin{itemize} 
\item $|A_{1}|=|A_{2}|=\frac{2^{n}+2^{\frac{n+s}{2}}}{K}-2^{\frac{n+s}{2}-1}, |A_{i}|=\frac{2^{n}+2^{\frac{n+s}{2}}}{K}, 3 \leq i \leq K$.
\item $|A_{K-1}|=|A_{K}|=\frac{2^{n}-2^{\frac{n+s}{2}}}{K}+2^{\frac{n+s}{2}-1}, |A_{i}|=\frac{2^{n}-2^{\frac{n+s}{2}}}{K}, 1 \leq i \leq K-2$.
\item $K$ is a power of $2$ (denote $K=2^m$), $|A_{1}|=2^{n-m}-2^{\frac{n+s}{2}-1}, |A_{K}|=2^{n-m}+2^{\frac{n+s}{2}-1}, |A_{i}|=2^{n-m}, 2 \leq i \leq K-1$.
\end{itemize}
When $p=3$, the cardinality of $A_{i}$ satisfies one of the following:
\begin{itemize}
\item $K=6$, $|A_{1}|=|A_{2}|=|A_{3}|=\frac{3^{n-1}-3^{\frac{n+s-1}{2}}}{2}, |A_{4}|=|A_{5}|=|A_{6}|=\frac{3^{n-1}+3^{\frac{n+s-1}{2}}}{2}$.
 \item $K$ is a power of $3$ (denote $K=3^m$), $|A_{1}|=3^{n-m}-3^{\frac{n+s-1}{2}},|A_{K}|=3^{n-m}+3^{\frac{n+s-1}{2}}, |A_{i}|=3^{n-m}, 2 \leq i \leq K-1$.
\end{itemize}
(ii) If the generated $p$-ary $s$-plateaued functions by $\Gamma$ are all balanced (called balanced plateaued partition), then $K$ is a power of $p$ (denote $K=p^m$) and $|A_{i}|=p^{n-m}, 1 \leq i \leq K$.

(iii) If the generated $p$-ary $s$-plateaued functions by $\Gamma$ are all unbalanced (called unbalanced plateaued partition), then $n+s$ is even. Furthermore, the cardinality of $A_{i}$ satisfies one of the following:
\begin{itemize}
	\item $|A_{1}|=\frac{p^{n}+p^{\frac{n+s}{2}}}{K}-p^{\frac{n+s}{2}}, |A_{i}|=\frac{p^{n}+p^{\frac{n+s}{2}}}{K}, 2 \leq i \leq K$.
	\item $|A_{K}|=\frac{p^{n}-p^{\frac{n+s}{2}}}{K}+p^{\frac{n+s}{2}}, |A_{i}|=\frac{p^{n}-p^{\frac{n+s}{2}}}{K}, 1 \leq i \leq K-1$.
\end{itemize}
\end{proposition}

\begin{proof}
$(i)$ We first show that when $p \geq 5$, or $p=3$ and $n+s$ is even, then the generated $p$-ary $s$-plateaued functions by $\Gamma$ are all balanced or all unbalanced. Define map $F: V_{n}^{(p)}\rightarrow \{1, 2, \dots, K\}$ as
\begin{equation*}
	F(x)=i, \text{ if } x \in A_{i}, 1\leq i \leq K.
\end{equation*}
By the definition of $s$-plateaued partitions, for any balanced map $B$ from $\{1, 2, \dots, K\}$ to $\mathbb{F}_{p}$ (that is, $|D_{B, j}|=\frac{K}{p}, j \in \mathbb{F}_{p}$), $B(F(x))$ is a $p$-ary $s$-plateaued function generated by $\Gamma$. Let $f, f'$ be any $p$-ary $s$-plateaued functions generated by $\Gamma$. Denote $D_{f, j}=\bigcup_{i=1}^{\frac{K}{p}}A_{b_{i}^{(j)}}$, and $D_{f', j}=\bigcup_{i=1}^{\frac{K}{p}}A_{{b^{'}_{i}}^{(j)}}$, where $j \in \mathbb{F}_{p}$. Then $f(x)=\widetilde{B}(F(x))$ and $f'(x)=\widetilde{B}'(F(x))$, where $\widetilde{B}$ and $\widetilde{B}'$ are the balanced maps from $\{1, 2, \dots, K\}$ to $\mathbb{F}_{p}$ defined by
\begin{equation*}
	\widetilde{B}(b_{i}^{(j)})=j, \widetilde{B}'({b^{'}_{i}}^{(j)})=j, 1\leq i \leq \frac{K}{p}, j \in \mathbb{F}_{p}.
\end{equation*}
Define the permutation $\widetilde{P}$ of $\{1, 2, \dots, K\}$ as
\begin{equation*}
	\widetilde{P}({b^{'}_{i}}^{(j)})=b_{i}^{(j)}, 1 \leq i \leq \frac{K}{p}, j \in \mathbb{F}_{p}.
\end{equation*}
Then $f'(x)=\widetilde{B}(\widetilde{P}(F(x)))$. We claim that if $\delta_{S_{B(P(F(x)))}}(0)=\delta_{S_{B(F(x))}}(0)$ for any balanced map $B$ from $\{1, 2, \dots, K\}$ to $\mathbb{F}_{p}$ and any transposition permutation $P$ of $\{1, 2, \dots, K\}$, then $\delta_{S_{f}}(0)=\delta_{S_{f'}}(0)$. Since any permutation can be decomposed into transposition permutations, then there are transposition permutations $P_{1}, P_{2}, \dots, P_{t}$ such that $\widetilde{P}=P_{1} P_{2} \cdots P_{t}$. If $t=1$, then the claim holds. We consider the case of $t\geq 2$. Define $B^{(i)}=\widetilde{B}(P_{1} \cdots P_{i}), 1\leq i \leq t-1$. Then $B^{(t-1)}, \dots, B^{(1)}$ are balanced maps from $\{1, 2, \dots, K\}$ to $\mathbb{F}_{p}$, and we obtain
\begin{equation*}
	\begin{split}
		\delta_{S_{f'}}(0)&=\delta_{S_{\widetilde{B}(P_{1} \cdots P_{t}(F(x)))}}(0)=\delta_{S_{B^{(t-1)}(P_{t}(F(x)))}}(0)=\delta_{S_{B^{(t-1)}(F(x))}}(0)\\
		&=\delta_{S_{B^{(t-2)}(P_{t-1}(F(x)))}}(0)=\delta_{S_{B^{(t-2)}(F(x))}}(0)=\cdots \\
		&=\delta_{S_{B^{(1)}(F(x))}}(0)=\delta_{S_{\widetilde{B}(P_{1}(F(x)))}}(0)=\delta_{S_{\widetilde{B}(F(x))}}(0)=\delta_{S_{f}}(0).
	\end{split}
\end{equation*}
Therefore, the claim holds. Let $B$ be any fixed balanced map from $\{1, 2, \dots, K\}$ to $\mathbb{F}_{p}$, and $P$ be any fixed transposition permutation of $\{1, 2, \dots, K\}$. Denote $P=(i_{0} i_{1})$, where $1 \leq i_{0}<i_{1} \leq K$. In the following, we only need to prove that $\delta_{S_{B(P(F(x)))}}(0)=\delta_{S_{B(F(x))}}(0)$ if $p \geq 5$, or $p=3$ and $n+s$ is even. We use the proof by contradiction. W.l.o.g, assume that $0 \in S_{B(P(F(x)))}, 0 \notin S_{B(F(x))}$. We have
\begin{equation}\label{p1 1}
	\begin{split}
		\epsilon_{B(P(F))}(0)p^{\frac{n+s}{2}}\zeta_{p}^{(B(P(F)))^{*}(0)}&=W_{B(P(F))}(0)\\
		&=\sum_{x \in V_{n}^{(p)}}\zeta_{p}^{B(P(F(x)))}=\sum_{i=1}^{K}\zeta_{p}^{B(P(i))}|A_{i}|\\
		&=\sum_{i \in \{1, 2, \dots, K\} \backslash \{i_{0}, i_{1}\}} \zeta_{p}^{B(i)}|A_{i}|+\zeta_{p}^{B(i_{1})}|A_{i_{0}}|+\zeta_{p}^{B(i_{0})}|A_{i_{1}}|,
	\end{split}
\end{equation}
and 
\begin{equation} \label{p1 2}
	\begin{split}
		0=W_{B(F)}(0)&=\sum_{x \in V_{n}^{(p)}}\zeta_{p}^{B(F(x))}=\sum_{i=1}^{K}\zeta_{p}^{B(i)}|A_{i}|\\
		&=\sum_{i \in \{1, 2, \dots, K\} \backslash \{i_{0}, i_{1}\}} \zeta_{p}^{B(i)}|A_{i}|+\zeta_{p}^{B(i_{0})}|A_{i_{0}}|+\zeta_{p}^{B(i_{1})}|A_{i_{1}}|.
	\end{split}
\end{equation}
Subtracting Eq. \eqref{p1 1} from Eq. \eqref{p1 2}, we get
\begin{equation}\label{p1 3}
\epsilon_{B(P(F))}(0)p^{\frac{n+s}{2}}\zeta_{p}^{(B(P(F)))^{*}(0)}=(\zeta_{p}^{B(i_{0})}-\zeta_{p}^{B(i_{1})})(|A_{i_{1}}|-|A_{i_{0}}|).
\end{equation}
Case I: $B(i_{0})=B(i_{1})$ or $|A_{i_{0}}|=|A_{i_{1}}|$. This case contradicts Eq. \eqref{p1 3}.\\
Case II: $B(i_{0}) \neq B(i_{1})$ and $|A_{i_{0}}| \neq |A_{i_{1}}|$. Denote $|A_{i_{1}}|-|A_{i_{0}}|=p^{t}r$, where $t \geq 0$ is an integer and $r$ is a positive integer with $p \nmid r$. Let $L=\mathbb{Q}(\zeta_{p})$ and  $\mathcal{O}_{L}=\mathbb{Z}[\zeta_{p}]$. Notice that $\epsilon_{B(P(F))}(0)p^{\frac{n+s}{2}} \in \mathcal{O}_{L}$ and $(\epsilon_{B(P(F))}(0)p^{\frac{n+s}{2}})^{2}=\pm p^{n+s}$. Since $p \mathcal{O}_{L}=((1-\zeta_{p})\mathcal{O}_{L})^{p-1}$, then $(\epsilon_{B(P(F))}(0)p^{\frac{n+s}{2}} \mathcal{O}_{L})^{2}=((1-\zeta_{p})\mathcal{O}_{L})^{(p-1)(n+s)}$. By the
uniqueness of the decomposition of $\epsilon_{B(P(F))}(0)p^{\frac{n+s}{2}} \mathcal{O}_{L}$, we have $\epsilon_{B(P(F))}(0)p^{\frac{n+s}{2}}\mathcal{O}_{L}=((1-\zeta_{p})\mathcal{O}_{L})^{\frac{(p-1)(n+s)}{2}}$. Then by Eq. \eqref{p1 3} and  $(\zeta_{p}^{B(i_{0})}-\zeta_{p}^{B(i_{1})})\mathcal{O}_{L}=(1-\zeta_{p}^{B(i_{1})-B(i_{0})})\mathcal{O}_{L}=(1-\zeta_{p})\mathcal{O}_{L}$, we obtain
\begin{equation}\label{p1 4}
	((1-\zeta_{p})\mathcal{O}_{L})^{\frac{(p-1)(n+s)}{2}}=((1-\zeta_{p})\mathcal{O}_{L})^{1+(p-1)t}
	\cdot r\mathcal{O}_{L}.
\end{equation}
Since $p \nmid r$, then $(1-\zeta_{p}) \mathcal{O}_{L} \nmid r\mathcal{O}_{L}$. Then by Eq. \eqref{p1 4}, we get $r=1$ and then $\frac{(p-1)(n+s)}{2}=1+(p-1)t$, which is impossible when $p \geq 5$, or $p=3$ and $n+s$ is even (since $\frac{p-1}{2}(n+s-2t)=1$ implies that $p=3$ and $n+s=2t+1$). Thus, when $p \geq 5$, or $p=3$ and $n+s$ is even, the generated $p$-ary $s$-plateaued functions by $\Gamma$ are all balanced or all unbalanced.

In the following, for the case of $p=2$, or $p=3, 2 \nmid (n+s)$, by using the value distribution of $p$-ary plateaued functions, we analyze the possible values of $|A_{i}|$ if $\Gamma$ generates both balanced and unbalanced $p$-ary $s$-plateaued functions simultaneously. Let $A_{i}, A_{j}$ be arbitrary fixed sets in $\Gamma$ with $|A_{i}| \neq |A_{j}|$ (such $A_{i}, A_{j}$ must exist otherwise the generated $p$-ary plateaued functions are all balanced). W.l.o.g, suppose $|A_{i}|>|A_{j}|$.

When $p=2$, let $f$ be a Boolean $s$-plateaued function generated by $\Gamma$ with $A_{i} \subseteq D_{f, 0}, A_{j} \subseteq D_{f, 1}$. Let $g$ be the Boolean $s$-plateaued function generated by $\Gamma$ with  $D_{g, 0}=A_{j} \cup D_{f, 0} \backslash A_{i}$. Then by the value distribution of Boolean $s$-plateaued functions, $|A_{i}|-|A_{j}|=|D_{f, 0}|-|D_{g, 0}|=2^{\frac{n+s}{2}-1}$ or $2^{\frac{n+s}{2}}$. By the arbitrariness of $A_{i}$ and $A_{j}$, only the following cases may occur:
\begin{itemize}
	\item[Case (1)] $|A_{i}|, 1 \leq i \leq K$, take two different values $X$ and $X+\varepsilon 2^{\frac{n+s}{2}}$, where the frequency of $X$ is greater than or equal to that of $X+\varepsilon 2^{\frac{n+s}{2}}$, and $\varepsilon \in \{\pm 1\}$ is a constant. 
	In fact, this case cannot occur. Note that the frequency of $X+\varepsilon 2^{\frac{n+s}{2}}$ is at least $2$ otherwise the generated Boolean plateaued functions are all unbalanced. Let $A_{i_{1}}, A_{i_{2}}, A_{i_{3}}, A_{i_{4}}$ be with $|A_{i_{1}}|=|A_{i_{2}}|=X, |A_{i_{3}}|=|A_{i_{4}}|=X+\varepsilon 2^{\frac{n+s}{2}}$. Let $f'$ be a Boolean $s$-plateaued function generated by $\Gamma$ with $A_{i_{1}} \cup A_{i_{2}} \subseteq D_{f', 0}$, $A_{i_{3}} \cup A_{i_{4}} \subseteq D_{f', 1}$ . Let $g'$ be the Boolean $s$-plateaued function generated by $\Gamma$ with $D_{g', 0}=A_{i_{3}} \cup A_{i_{4}} \cup D_{f', 0} \backslash (A_{i_{1}} \cup A_{i_{2}})$. Then $|D_{f', 0}|-|D_{g', 0}|=\pm 2^{\frac{n+s}{2}+1}$, which contradicts the value distribution of Boolean plateaued functions.
	\item[Case (2)] $|A_{i}|, 1 \leq i \leq K$, take two different values $X$ and $X+\varepsilon 2^{\frac{n+s}{2}-1}$, where the frequency of $X$ is greater than or equal to that of $X+\varepsilon 2^{\frac{n+s}{2}-1}$, and $\varepsilon \in \{\pm 1\}$ is a constant. The frequency of $X+\varepsilon 2^{\frac{n+s}{2}-1}$ is at least $2$ otherwise the generated Boolean plateaued functions are all unbalanced. We show the frequency of $X+\varepsilon 2^{\frac{n+s}{2}-1}$ is $2$. To the contrary, assume its frequency is at least $3$. Then $K \geq 6$ and $\frac{K}{2} \geq 3$. Let $A_{i_{t}}, 1
	\leq t \leq 6$, be with $|A_{i_{1}}|=|A_{i_{2}}|=|A_{i_{3}}|=X, |A_{i_{4}}|=|A_{i_{5}}|=|A_{i_{6}}|=X+\varepsilon 2^{\frac{n+s}{2}-1}$. Let $f''$ be a Boolean $s$-plateaued function generated by $\Gamma$ with $A_{i_{1}} \cup A_{i_{2}} \cup A_{i_{3}} \subseteq D_{f'', 0}, A_{i_{4}} \cup A_{i_{5}} \cup A_{i_{6}} \subseteq D_{f'', 1}$. 
	Let $g''$ be the Boolean $s$-plateaued function generated by $\Gamma$ with $D_{g'', 0}=A_{i_{4}} \cup A_{i_{5}} \cup A_{i_{6}} \cup D_{f'', 0} \backslash (A_{i_{1}} \cup A_{i_{2}} \cup A_{i_{3}})$. Then $|D_{f'', 0}|-|D_{g'', 0}|=\pm 3 \cdot 2^{\frac{n+s}{2}-1}$, which contradicts the value distribution of Boolean plateaued functions. Since the frequency of $X+\varepsilon 2^{\frac{n+s}{2}-1}$ is $2$, then by the value distribution of Boolean $s$-plateaued functions, we get $X=\frac{2^{n}-\varepsilon 2^{\frac{n+s}{2}}}{K}$, and the corresponding values $|A_{i}|, 1 \leq i \leq K$, can be obtained.
	\item[Case (3)] $|A_{i}|, 1 \leq i \leq K$, take three different values $X$, $X+2^{\frac{n+s}{2}-1}$ and $X+2^{\frac{n+s}{2}}$. We claim that the frequency of $X$ and $X+2^{\frac{n+s}{2}}$ are both $1$. To the contrary, assume that the frequency of $X$ or $X+2^{\frac{n+s}{2}}$ is at least $2$. We only consider the case of the frequency of $X$ is at least $2$ since the other case is similar. Let $A_{i_{1}}, A_{i_{2}}, A_{i_{3}}, A_{i_{4}}$ be with $|A_{i_{1}}|=|A_{i_{2}}|=X, |A_{i_{3}}|=X+2^{\frac{n+s}{2}-1},|A_{i_{4}}|=X+2^{\frac{n+s}{2}}$. Let $f'''$ be a Boolean $s$-plateaued function generated by $\Gamma$ with $A_{i_{1}} \cup A_{i_{2}} \subseteq D_{f''', 0}$, $A_{i_{3}} \cup A_{i_{4}} \subseteq D_{f''', 1}$. Let $g'''$ be the Boolean $s$-plateaued function generated by $\Gamma$ with $D_{g''', 0}=A_{i_{3}} \cup A_{i_{4}} \cup D_{f''', 0} \backslash (A_{i_{1}} \cup A_{i_{2}})$. Then $|D_{g''', 0}|-|D_{f''', 0}|=3 \cdot 2^{\frac{n+s}{2}-1}$, which contradicts the value distribution of Boolean plateaued functions. Thus, the claim holds. Since the frequency of $X$ and $X+2^{\frac{n+s}{2}}$ are both $1$, then by the value distribution of Boolean $s$-plateaued functions, we get $X=\frac{2^n}{K}-2^{\frac{n+s}{2}-1}$. As $X$ is an integer, $K=2^m$ for some positive integer. Then the corresponding values $|A_{i}|, 1 \leq i \leq K$, can be obtained.
\end{itemize}

When $p=3$ and $n+s$ is odd, let $f$ be a ternary $s$-plateaued function generated by $\Gamma$ with $A_{i} \subseteq D_{f, 0}, A_{j} \subseteq D_{f, 1}$. Let $g$ be a ternary $s$-plateaued function with $D_{g, 0}=A_{j} \cup D_{f, 0} \backslash A_{i}$. Then by the value distribution of ternary $s$-plateaued functions, $|A_{i}|-|A_{j}|=|D_{f, 0}|-|D_{g, 0}|=3^{\frac{n+s-1}{2}}$ or $2 \cdot 3^{\frac{n+s-1}{2}}$. By the arbitrariness of $A_{i}$ and $A_{j}$, only the following cases may occur:
\begin{itemize}
	\item[Case (1)] $|A_{i}|, 1 \leq i \leq K$, take two different values $X$ and $X+ \varepsilon 2 \cdot 3^{\frac{n+s-1}{2}}$, where the frequency of $X$ is greater than or equal to that of $X+\varepsilon 2 \cdot 3^{\frac{n+s-1}{2}}$, and $\varepsilon \in \{\pm 1\}$ is a constant. 
	With the similar arguments as for $p=2$, it is easy to show that this case cannot occur.
	\item[Case (2)] $|A_{i}|, 1 \leq i \leq K$, take two different values $X$ and $X+ \varepsilon 3^{\frac{n+s-1}{2}}$, where the frequency of $X$ is greater than or equal to that of $X+\varepsilon 3^{\frac{n+s-1}{2}}$, and $\varepsilon \in \{\pm 1\}$ is a constant. The frequency of $X+\varepsilon 3^{\frac{n+s-1}{2}}$ is at least $2$ otherwise the generated ternary plateaued functions are all unbalanced. When $\frac{K}{3} \geq 3$, with the similar arguments as for $p=2$, it is easy to show that the possible frequency of $X+\varepsilon 3^{\frac{n+s-1}{2}}$ can only be $2$. When $\frac{K}{3} \geq 3$, or $K=6$ and the frequency of $X+\varepsilon 3^{\frac{n+s-1}{2}}$ is $2$, let $A_{i_{1}}, A_{i_{2}}$ be with $|A_{i_{1}}|=|A_{i_{2}}|=X+\varepsilon 3^{\frac{n+s-1}{2}}$. Then $|A_{i}|=X, i \neq i_{1}, i_{2}$. Let $f''$ be a ternary $s$-plateaued function generated by $\Gamma$ with $A_{i_{1}} \subseteq D_{f'', 0}$, $A_{i_{2}} \subseteq D_{f'', 1}$. Then $|D_{f'', 0}|=|D_{f'', 1}| \neq |D_{f'', 2}|$, which contradicts the value distribution of ternary plateaued functions when $n+s$ is odd. Thus, $K=6$ and the frequency of $X+\varepsilon 3^{\frac{n+s-1}{2}}$ is $3$. By the value distribution of ternary $s$-plateaued functions, we get $X=\frac{3^{n-1}-\varepsilon 3^{\frac{n+s-1}{2}}}{2}$, and the corresponding values $|A_{i}|, 1 \leq i \leq 6$, can be obtained.
	\item[Case (3)] $|A_{i}|, 1 \leq i \leq K$, take three different values $X$, $X+3^{\frac{n+s-1}{2}}$ and $X+2 \cdot 3^{\frac{n+s-1}{2}}$. With the similar arguments as for $p=2$, the frequency of $X$ and  $X+2 \cdot 3^{\frac{n+s-1}{2}}$ are both 1. Then by the value distribution of ternary $s$-plateaued functions, we get $X=\frac{3^n}{K}-3^{\frac{n+s-1}{2}}$. As X is an integer, $K=3^m$ for some positive integer. Then the corresponding values $|A_{i}|, 1 \leq i \leq K$, can be obtained.
\end{itemize}

$(ii)$ If the generated $p$-ary $s$-plateaued functions are all balanced, then the sum of the cardinality of any $\frac{K}{p}$ of sets $A_{i}$ are $p^{n-1}$. Hence, $|A_{i}|, 1 \leq i \leq K$, are the same and $|A_{i}|=\frac{p^n}{K}$. Since $|A_{i}|$ is an integer, then $K=p^m$ for some positive integer $m$.

$(iii)$ Obviously, $n+s$ is even if $p=2$. Recall that $K \neq 3$. In the following, we show that $n+s$ is also even if $p$ is odd. We first prove that when $p \geq 5$ and $n+s$ is odd, there is no unbalanced $s$-plateaued partition of $V_{n}^{(p)}$. Let $f$ be a $p$-ary $s$-plateaued function generated by $\Gamma$. Let $\alpha, \beta \in \mathbb{F}_{p} \backslash \{f^{*}(0)\}$ with $\eta_{1}(f^{*}(0)-\alpha) \neq \eta_{1}(f^{*}(0)-\beta)$. Define the $p$-ary $s$-plateaued function $g$ as $D_{g, \alpha}=D_{f, \beta}, D_{g, \beta}=D_{f, \alpha}, D_{g, i}=D_{f, i}, i \in \mathbb{F}_{p} \backslash \{\alpha, \beta\}$. Since $|D_{g, i}|=|D_{f, i}|, i \in \mathbb{F}_{p} \backslash \{\alpha, \beta\}$, then by the value distribution of $p$-ary plateaued functions, we have that $f^{*}(0)=g^{*}(0), \epsilon_{f}(0)=\epsilon_{g}(0)$. As $\eta_{1}(f^{*}(0)-\alpha) \neq \eta_{1}(f^{*}(0)-\beta)$, by the value distribution of $p$-ary plateaued functions, $|D_{f, \alpha}| \neq |D_{g, \beta}|$, which is a contradiction. Now we prove that when $p=3, 2 \nmid (n+s), K \geq 6$, there is no unbalanced $s$-plateaued partition of $V_{n}^{(3)}$ of depth $K$. Let $A_{i_{1}}, A_{i_{2}}, A_{i_{3}}$ be arbitrary sets in $\Gamma$. W.l.o.g., suppose $|A_{i_{1}}| \geq |A_{i_{2}}| \geq |A_{i_{3}}|$. Let $\{C_{1}, C_{2}, C_{3}\}$ be a partition of $V_{n}^{(3)} \backslash (A_{i_{1}} \cup A_{i_{2}} \cup A_{i_{3}})$ for which $C_{j} \ (j=1, 2, 3)$ is the union of $\frac{K}{3}-1$ of sets $A_{i}$ in $\Gamma$. W.l.o.g, suppose $|C_{1}| \geq |C_{2}| \geq |C_{3}|$. By $|A_{i_{1}}|+|C_{1}| \geq |A_{i_{2}}|+|C_{2}| \geq |A_{i_{3}}|+|C_{3}|$ and the value distribution of ternary $s$-plateaued functions, we have that $|A_{i_{1}}|+|C_{1}|=3^{n-1}+3^{\frac{n+s-1}{2}}$, $|A_{i_{2}}|+|C_{2}|=3^{n-1}$, $|A_{i_{3}}|+|C_{3}|=3^{n-1}-3^{\frac{n+s-1}{2}}$. Let $f$ be the ternary $s$-plateaued function generated by $\Gamma$ with $D_{f, 0}=A_{i_{1}} \cup C_{3}$, $D_{f, 1}=A_{i_{2}} \cup C_{2}$ and $D_{f, 2}=A_{i_{3}} \cup C_{1}$. Then $|D_{f, 1}|=3^{n-1}$ and thus $|A_{i_{1}}|+|C_{3}|=|D_{f, 0}| \neq 3^{n-1}$ by the value distribution of ternary plateaued functions. It is easy to see that $|A_{i_{1}}|+|C_{3}|=3^{n-1}-3^{\frac{n+s-1}{2}}$ implies $|A_{i_{1}}|=|A_{i_{2}}|=|A_{i_{3}}|$, and $|A_{i_{1}}|+|C_{3}|=3^{n-1}+3^{\frac{n+s-1}{2}}$ implies $|C_{1}|=|C_{2}|=|C_{3}|$ and then $|A_{i_{1}}|>|A_{i_{2}}|>|A_{i_{3}}|$. Further, by the arbitrariness of $A_{i_{1}}, A_{i_{2}}, A_{i_{3}}$, we have that $|A_{K}|>|A_{K-1}|>\cdots>|A_{1}|$, which is impossible by $K \geq 6$ and the value distribution of ternary $s$-plateaued functions. Therefore, $n+s$ is even.

In the following, we compute the possible values of $|A_{r}|, 1 \leq r \leq K$. Let $A_{i}, A_{j}$ be arbitrary fixed sets in $\Gamma$ with $|A_{i}|\neq |A_{j}|$. W.o.l.g., suppose $|A_{i}|>|A_{j}|$. Let $f$ be a $p$-ary $s$-plateaued function generated by $\Gamma$ with $A_{i} \subseteq D_{f, 0}$, $A_{j} \subseteq D_{f, 1}$. Let $g$ be the $p$-ary $s$-plateaued function generated by $\Gamma$ for which $D_{g, 0}=A_{j} \cup D_{f, 0} \backslash A_{i}$, $D_{g, 1}=A_{i} \cup D_{f, 1} \backslash A_{j}$, and when $p \geq 3$, for $k \neq i, j$, $D_{g, k}=D_{f, k}$. By the value distribution of unbalanced $p$-ary $s$-plateaued functions, we obtain $|A_{i}|-|A_{j}|=|D_{f, 0}|-|D_{g, 0}|=p^{\frac{n+s}{2}}$. By the arbitrariness of $A_{i}, A_{j}$, we have that the values of $|A_{r}|, 1 \leq r \leq K$, take two different values $X$ and $X+\varepsilon p^{\frac{n+s}{2}}$, where the frequency of $X$ is greater than or equal to that of $X+\varepsilon p^{\frac{n+s}{2}}$, and $\varepsilon \in \{\pm 1\}$ is a constant. Using the value distribution of unbalanced $p$-ary $s$-plateaued functions, we get that the frequency of $X+\varepsilon p^{\frac{n+s}{2}}$ is $1$, and $X=\frac{p^n-\varepsilon p^{\frac{n+s}{2}}}{K}$. Then the corresponding values $|A_{r}|, 1 \leq r \leq K$, can be obtained.
\end{proof}

\section{Constructions of $s$-plateaued partitions}
In this section, we will present some constructions of $s$-plateaued partitions for which any generated $p$-ary $s$-plateaued function has no nonzero linear structure.

The following proposition provides a general way to obtain $s$-plateaued partitions. 

\begin{proposition} \label{Proposition 2}
Suppose $F: V_{n}^{(p)} \rightarrow V_{m}^{(p)}$ is a vectorial $s$-plateaued function (where $p^m>3$) satisfying the following conditions:
\begin{itemize}
	\item The Walsh supports of all component functions $S_{F_{c}}, c \in {V_{m}^{(p)}}^{*}$, are all the same (denoted by $S$).
	\item Each component function $F_{c}$ ($c \in {V_{m}^{(p)}}^{*}$) satisfies $\epsilon_{F_{c}}(a)=\epsilon(a)$ for any $a \in S$, where $\epsilon(a)$ is independent of $c$. 
\end{itemize}
Then $\{D_{F, i}, i \in V_{m}^{(p)}\}$ is an $s$-plateaued partition if any one of the following holds.

(i) There exist functions $G: S \rightarrow V_{m}^{(p)}$ and $g: S \rightarrow \mathbb{F}_{p}$ such that 
\begin{equation}\label{Pr2 0}
	(F_{c})^{*}(x)=G_{c}(x)+g(x), x \in S, c \in  {V_{m}^{(p)}}^{*}.
\end{equation}

(ii) $(F_{d})^{*}(x)+(F_{e})^{*}(x)-(F_{d+e})^{*}(x), d \neq -e \in  {V_{m}^{(p)}}^{*}$, are the same function from $S$ to $\mathbb{F}_{p}$.
\end{proposition}

\begin{proof}
We first show that if $(i)$ holds, then $\{D_{F, i}, i \in V_{m}^{(p)}\}$ is an $s$-plateaued partition. For any $a \in V_{n}^{(p)}$ and $i \in V_{m}^{(p)}$, by \cite[Proposition 3]{WF2023Ne},
	\begin{equation*}
		\chi_{-a}(D_{F, i})=p^{n-m}\delta_{0}(a)+p^{-m}\sum_{c \in {V_{m}^{(p)}}^{*}}W_{F_{c}}(a)\zeta_{p}^{-\langle c, i\rangle_{m}}.
	\end{equation*}
Let $B: V_{m}^{(p)} \rightarrow \mathbb{F}_{p}$ be any fixed balanced function.
If $a \notin S$, then $\chi_{-a}(D_{F, i})=p^{n-m}\delta_{0}(a)$. In this case, we have
\begin{equation*}	
W_{B(F)}(a)=\sum_{x \in V_{n}^{(p)}}\zeta_{p}^{B(F(x))-\langle a, x\rangle _{n}}=\sum_{i \in V_{m}^{(p)}}\zeta_{p}^{B(i)}\chi_{-a}(D_{F, i})=
p^{n-m}\delta_{0}(a)\sum_{i \in V_{m}^{(p)}}\zeta_{p}^{B(i)}=0.
\end{equation*}

If $a \in S$, then 
\begin{equation*}
	\begin{split}
		\chi_{-a}(D_{F, i})&=p^{n-m}\delta_{0}(a)+\epsilon(a) p^{\frac{n+s}{2}-m}\sum_{c \in {V_{m}^{(p)}}^{*}}\zeta_{p}^{(F_{c})^{*}(a)-\langle c, i\rangle_{m}}\\
		&=p^{n-m}\delta_{0}(a)+\epsilon(a) p^{\frac{n+s}{2}-m}\zeta_{p}^{g(a)}\sum_{c \in {V_{m}^{(p)}}^{*}}\zeta_{p}^{\langle c, G(a)-i\rangle_{m}}\\
		&=p^{n-m}\delta_{0}(a)+\epsilon(a) p^{\frac{n+s}{2}-m}\zeta_{p}^{g(a)}(p^m \delta_{i}(G(a))-1).
	\end{split}
\end{equation*}
In this case, we have
\begin{equation*}	
	\begin{split}
	W_{B(F)}(a)&=\sum_{i \in V_{m}^{(p)}}\zeta_{p}^{B(i)}\chi_{-a}(D_{F, i})\\
	&=(p^{n-m}\delta_{0}(a)-\epsilon(a) p^{\frac{n+s}{2}-m}\zeta_{p}^{g(a)})\sum_{i \in V_{m}^{(p)}}\zeta_{p}^{B(i)}+\epsilon(a) p^{\frac{n+s}{2}}\zeta_{p}^{B(G(a))+g(a)}\\
	&=\epsilon(a) p^{\frac{n+s}{2}}\zeta_{p}^{B(G(a))+g(a)}.
	\end{split}
\end{equation*}
Therefore, $B(F(x))$ is an $s$-plateaued function, and by the definition of $s$-plateaued partitions, $\{D_{F, i}, i \in V_{m}^{(p)}\}$ is an $s$-plateaued partition.

In the following, we show that $(i)$ and $(ii)$ are equivalent. Obviously, $(i)$ implies $(ii)$. Suppose that $(ii)$ holds. Then there is $g: S \rightarrow \mathbb{F}_{p}$ such that
\begin{equation} \label{Pr2 1}
	(F_{d})^{*}(x)-(F_{e})^{*}(x)-(F_{d-e})^{*}(x)=-g(x), x \in S, d \neq e \in {V_{m}^{(p)}}^{*}.
\end{equation} 
Let $\{\beta_{1}, \dots, \beta_{m}\}$ be a basis of $V_{m}^{(p)}$ over $\mathbb{F}_{p}$. For any $x \in S$, let $G(x) \in V_{m}^{(p)}$ be defined by the equations: $\langle \beta_{i}, G(x)\rangle_{m}=(F_{\beta_{i}})^{*}(x)-g(x), i=1, 2, \dots, m$. For any $c \in {V_{m}^{(p)}}^{*}$, denote $c=c_{i_{1}}\beta_{i_{1}}+\dots+c_{i_{t}}\beta_{i_{t}}$, where $c_{i_{j}} \in \mathbb{F}_{p}^{*}, 1\leq j \leq t, 1\leq i_{1} < \dots <i_{t}\leq m, 1\leq t\leq m$. Then
\begin{equation} \label{Pr2 2}
	G_{c}(x)=\sum_{j=1}^{t}c_{i_{j}}\langle \beta_{i_{j}}, G(x)\rangle_{m}=\sum_{j=1}^{t}c_{i_{j}}(F_{\beta_{i_{j}}})^{*}(x)-g(x)\sum_{j=1}^{t}c_{i_{j}}, x \in S.
\end{equation}
	For any $a \in \mathbb{F}_{p}^{*}$ and $c \in {V_{m}^{(p)}}^{*}$, we first show
	\begin{equation} \label{Pr2 3}
		(F_{ac})^{*}(x)=a(F_{c})^{*}(x)-(a-1)g(x), x \in S.
	\end{equation}
	Obviously, when $a=1$, Eq. \eqref{Pr2 3} holds. Suppose Eq. \eqref{Pr2 3} holds for $a=k$, where $1\leq k \leq p-2$. When $a=k+1$, by Eq. \eqref{Pr2 1}, 
	\begin{equation*}
		\begin{split}
			(F_{(k+1)c})^{*}(x)&=(F_{kc})^{*}(x)+(F_{c})^{*}(x)-g(x)=k(F_{c})^{*}(x)-(k-1)g(x)+(F_{c})^{*}(x)-g(x)\\
			&=(k+1)(F_{c})^{*}(x)-kg(x).
		\end{split}
	\end{equation*}
Then Eq. \eqref{Pr2 3} also holds for $a=k+1$. By induction, Eq. \eqref{Pr2 3} holds. When $t=1$, by Eqs. \eqref{Pr2 2} and \eqref{Pr2 3}, Eq. \eqref{Pr2 0} holds. Suppose Eq. \eqref{Pr2 0} holds for $t=r$. When $t=r+1$, denote $c'=\sum_{j=1}^{r}c_{i_{j}}\beta_{i_{j}}$, by Eq. \eqref{Pr2 1}, \begin{equation*}
	\begin{split}
		(F_{c})^{*}(x)&=(F_{c'+c_{i_{r+1}}\beta_{i_{r+1}}})^{*}(x)=(F_{c'})^{*}(x)+(F_{c_{i_{r+1}}\beta_{i_{r+1}}})^{*}(x)-g(x)\\
		&=G_{c'}(x)+g(x)+G_{c_{i_{r+1}}\beta_{i_{r+1}}}(x)+g(x)-g(x)=G_{c}(x)+g(x).
	\end{split}
\end{equation*}
Then Eq. \eqref{Pr2 0} also holds for $t=r+1$. By induction, Eq. \eqref{Pr2 0} holds. This completes the proof.
\end{proof}

\begin{remark} \label{Remark 1}
Keep the same notations as in Proposition \ref{Proposition 2}. If $\{D_{F, i}, i \in V_{m}^{(p)}\}$ is an $s$-plateaued partition obtained by Proposition \ref{Proposition 2}, then by the proof of Proposition \ref{Proposition 2}, for any generated $s$-plateaued function $f(x)=B(F(x))$ (where $B$ is any balanced function from $V_{m}^{(p)}$ to $\mathbb{F}_{p}$), the Walsh support $S_{f}=S$ and $\epsilon_{f}(x)=\epsilon(x)$, $f^{*}(x)=B(G(x))+g(x), x \in S$.
\end{remark}

In the following, we will give some explicit constructions of vectorial $s$-plateaued functions satisfying the conditions in Proposition \ref{Proposition 2}. With these constructions, $s$-plateaued partitions can be obtained. First, we give the following lemma.

\begin{lemma}\label{new1}
Let $r, t, m$ be positive integers with $p^r>3, t<m$. Define $F: \mathbb{F}_{p^m} \times \mathbb{F}_{p^t} \rightarrow \mathbb{F}_{p^r}$ as 
$F(x, y)=\phi(x, y)+H(y)$, where $H: \mathbb{F}_{p^t} \rightarrow \mathbb{F}_{p^r}$, and $\phi: \mathbb{F}_{p^m} \times \mathbb{F}_{p^t} \rightarrow \mathbb{F}_{p^r}$ satisfies that for any given $y \in \mathbb{F}_{p^t}$, the map $x\mapsto \phi(x, y)$ is linear. For any $c \in \mathbb{F}_{p^r}^{*}$, let $\pi^{[c]}: \mathbb{F}_{p^t} \rightarrow \mathbb{F}_{p^m}$ be given by 
\begin{equation*}
	\mathrm{Tr}_{1}^{r}(c\phi(x, y))=\mathrm{Tr}_{1}^{m}(\pi^{[c]}(y)x), x \in \mathbb{F}_{p^m}, y \in \mathbb{F}_{p^t}.
\end{equation*}
Assume that the following conditions are satisfied:
\begin{itemize}
	\item For any $c \in \mathbb{F}_{p^r}^{*}$, $\pi^{[c]}$ is an injection, $\pi^{[c]}(\mathbb{F}_{p^t}), c \in \mathbb{F}_{p^r}^{*}$, are the same (denoted by $S_{1}$), and $(\pi^{[d]})^{-1}(x)+(\pi^{[e]})^{-1}(x)=(\pi^{[d+e]})^{-1}(x)$ for any $d \neq -e \in \mathbb{F}_{p^r}^{*}$ and $x \in S_{1}$;
	\item $\mathrm{Tr}_{1}^{r}(cH((\pi^{[c]})^{-1}(x))), c \in \mathbb{F}_{p^r}^{*}$, are the same function from $S_{1}$ to $\mathbb{F}_{p}$.
\end{itemize}
Then $\{D_{F, i}, i \in \mathbb{F}_{p^r}\}$ is an $(m-t)$-plateaued partition.
\end{lemma}

\begin{proof}
For any fixed $c \in \mathbb{F}_{p^r}^{*}$ and $y \in \mathbb{F}_{p^t}$, the map $x\mapsto \mathrm{Tr}_{1}^{r}(c\phi(x, y))$ is linear from $\mathbb{F}_{p^m}$ to $\mathbb{F}_{p}$. Since each linear map from $\mathbb{F}_{p^m}$ to $\mathbb{F}_{p}$ is of the form $\mathrm{Tr}_{1}^{m}(ax), x \in \mathbb{F}_{p^m}$, where $a \in \mathbb{F}_{p^m}$, then for any $c\in \mathbb{F}_{p^r}^{*}$ and $y \in \mathbb{F}_{p^t}$, there exists $\pi^{[c]}(y)\in \mathbb{F}_{p^m}$ such that
$\mathrm{Tr}_{1}^{r}(c\phi(x, y))=\mathrm{Tr}_{1}^{m}(\pi^{[c]}(y)x)$ for all $x \in \mathbb{F}_{p^m}$. For any $c \in \mathbb{F}_{p^r}^{*}$, $F_{c}(x, y)=\mathrm{Tr}_{1}^{r}(c\phi(x, y))+H_{c}(y)=\mathrm{Tr}_{1}^{m}(\pi^{[c]}(y)x)+H_{c}(y)$. For $a=(a_{1}, a_{2}) \in \mathbb{F}_{p^m} \times \mathbb{F}_{p^t}$, we have
\begin{equation*}
	\begin{split}
		W_{F_{c}}(a)&=\sum_{x \in \mathbb{F}_{p^m}, y \in \mathbb{F}_{p^t}}\zeta_{p}^{\mathrm{Tr}_{1}^{m}(\pi^{[c]}(y)x)+H_{c}(y)-\mathrm{Tr}_{1}^{m}(a_{1}x)-\mathrm{Tr}_{1}^{t}(a_{2}y)}\\
		&=\sum_{y \in \mathbb{F}_{p^t}}\zeta_{p}^{H_{c}(y)-\mathrm{Tr}_{1}^{t}(a_{2}y)}\sum_{x \in \mathbb{F}_{p^m}}\zeta_{p}^{\mathrm{Tr}_{1}^{m}((\pi^{[c]}(y)-a_{1})x)}\\
		&=\begin{cases}
			0, & \text{ if } a_{1} \notin S_{1},\\
			p^{m}\zeta_{p}^{-\mathrm{Tr}_{1}^{t}(a_{2}(\pi^{[c]})^{-1}(a_{1}))+H_{c}((\pi^{[c]})^{-1}(a_{1}))}, & \text{ if } a_{1} \in S_{1}.
		\end{cases}
	\end{split}
\end{equation*}
Then $F$ is a vectorial $(m-t)$-plateaued function for which $S_{F_{c}}=S_{1} \times \mathbb{F}_{p^t}$, $\epsilon_{F_{c}}(x, y)=1$, $(F_{c})^{*}(x, y)=-\mathrm{Tr}_{1}^{t}(y(\pi^{[c]})^{-1}(x))+\mathrm{Tr}_{1}^{r}(cH((\pi^{[c]})^{-1}(x))), x \in S_{1}, y \in \mathbb{F}_{p^t}$. Since $(\pi^{[d]})^{-1}(x)+(\pi^{[e]})^{-1}(x)=(\pi^{[d+e]})^{-1}(x)$ for any $d \neq -e \in \mathbb{F}_{p^r}^{*}$ and $x \in S_{1}$, and $\mathrm{Tr}_{1}^{r}(cH((\pi^{[c]})^{-1}(x))), c \in \mathbb{F}_{p^r}^{*}$, are the same function from $S_{1}$ to $\mathbb{F}_{p}$, then $(F_{d})^{*}(x, y)+(F_{e})^{*}(x, y)-(F_{d+e})^{*}(x, y), d \neq -e \in \mathbb{F}_{p^r}^{*}$, are the same function from $S_{1} \times \mathbb{F}_{p^t}$ to $\mathbb{F}_{p}$. By Proposition \ref{Proposition 2},  $\{D_{F, i}, i \in \mathbb{F}_{p^r}\}$ is an $(m-t)$-plateaued partition. 
\end{proof}

Based on Lemma \ref{new1}, by utilizing presemifields, we have the following construction of $s$-plateaued partitions.

\begin{theorem}\label{Theorem 0}
Let $r, t, m, u$ be positive integers with $p^r>3, r \mid t, t \mid m, r \neq t, t \neq m$, $u=p^{j_{0}}$ for some integer $0 \leq j_{0} \leq r-1$. Let $\mathbb{F}_{p^t}^{*}/\mathbb{F}_{p^r}^{*}=\{\bar{\beta}_{1}, \dots, \bar{\beta}_{\frac{p^t-1}{p^r-1}}\}$, and let $\bar{\alpha}_{i} \in \mathbb{F}_{p^m}^{*}/\mathbb{F}_{p^r}^{*}, 1 \leq i \leq \frac{p^t-1}{p^r-1}$. Let $\pi$ be a permutation of $\mathbb{F}_{p^m}$ satisfying $\pi(0)=0, \pi(c\beta_{i})=c^{u}\alpha_{\tau(i)}, 1 \leq i \leq \frac{p^t-1}{p^r-1}, c \in \mathbb{F}_{p^r}^{*}$, where $\tau$ is a permutation of $\{i, 1 \leq i \leq \frac{p^t-1}{p^r-1}\}$. Let $P=(\mathbb{F}_{p^m}, +, \circ)$ be a presemifield whose transpose $P^{T}=(\mathbb{F}_{p^m}, +, \star)$ is right $\mathbb{F}_{p^r}$-linear. Define $\psi: \mathbb{F}_{p^m} \times \mathbb{F}_{p^m} \rightarrow \mathbb{F}_{p^m}$ by $\psi(x, 0)=0$ and $\psi(x, y) \circ \pi(y)=x$ if $y \in \mathbb{F}_{p^m}^{*}$. Let $H: \mathbb{F}_{p^t} \rightarrow \mathbb{F}_{p^r}$ satisfy $H(0)=0$ and $H(c\beta_{i})=c^{-u}H(\beta_{i})$ for any $1 \leq i \leq \frac{p^t-1}{p^r-1}, c \in \mathbb{F}_{p^r}^{*}$. Then $\{D_{F, i}, i \in \mathbb{F}_{p^r}\}$ is an $(m-t)$-plateaued partition, where $F: \mathbb{F}_{p^m} \times \mathbb{F}_{p^t} \rightarrow \mathbb{F}_{p^r}$ is defined as 
\begin{equation*}
	F(x,y)=\mathrm{Tr}_{r}^{m}(\psi(x, y))+H(y). 
\end{equation*}
\end{theorem}

\begin{proof}
Let $\phi: \mathbb{F}_{p^m} \times \mathbb{F}_{p^t} \rightarrow \mathbb{F}_{p^r}$ be defined by  $\phi(x, y)=\mathrm{Tr}_{r}^{m}(\psi(x, y)), x \in \mathbb{F}_{p^m}, y \in \mathbb{F}_{p^t}$. When $y \in \mathbb{F}_{p^t}^{*}$, for any $x_{1}, x_{2} \in \mathbb{F}_{p^m}$, from $(\psi(x_{1}, y)+\psi(x_{2}, y))\circ \pi(y)=\psi(x_{1}, y) \circ \pi(y)+ \psi(x_{2}, y) \circ \pi(y)=x_{1}+x_{2}=\psi(x_{1}+x_{2}, y) \circ \pi(y)$ and $\pi(y) \neq 0$, we have $\psi(x_{1}+x_{2}, y)=\psi(x_{1}, y)+\psi(x_{2}, y)$. When $y=0$, then $\psi(x_{1}+x_{2}, 0)=\psi(x_{1}, 0)+\psi(x_{2}, 0)=0$. Thus, for any $y \in \mathbb{F}_{p^t}$, $x \mapsto \phi(x, y)$ is linear. In the following, we determine $\pi^{[c]}: \mathbb{F}_{p^t} \rightarrow \mathbb{F}_{p^m}$ with $\mathrm{Tr}_{1}^{r}(c\phi(x, y))=\mathrm{Tr}_{1}^{m}(\pi^{[c]}(y)x), x \in \mathbb{F}_{p^m}, y \in \mathbb{F}_{p^t}$. When $x=0$, let $\lambda_{x}=0$; when $x \in \mathbb{F}_{p^m}^{*}$, let $\lambda_{x} \in \mathbb{F}_{p^m}^{*}$ be given by $x \star \lambda_{x}^{-1}=1$. Define $\rho: \mathbb{F}_{p^m} \rightarrow \mathbb{F}_{p^m}$ as $\rho(x)=\pi^{-1}(\lambda_{x}^{-1})$. Then $\rho$ is a permutation of $\mathbb{F}_{p^m}$ with $\rho(0)=0$. For any fixed $c \in \mathbb{F}_{p^r}^{*}$ and $x \in \mathbb{F}_{p^m}^{*}$, let $z_{c, x}=\rho^{-1}(c^{-1}x)$. Then $\pi^{-1}(\lambda_{z_{c, x}}^{-1})=\rho(z_{c, x})=c^{-1}x$ and $\lambda_{z_{c, x}}^{-1}=\pi(c^{-1}x)$. For $y=d\beta_{i} \in \mathbb{F}_{p^t}^{*}$, where $d \in \mathbb{F}_{p^r}^{*}$ and $1 \leq i \leq \frac{p^t-1}{p^r-1}$, we have $\pi(c^{-1}y)=\pi(c^{-1}d\beta_{i})=c^{-u}d^{u}\alpha_{\tau(i)}=c^{-u}\pi(d \beta_{i})=c^{-u}\pi(y)$, and thus $\lambda_{z_{c, y}}^{-1}=c^{-u}\pi(y)$. Let $v=p^{r-j_{0}}$. For any $c \in \mathbb{F}_{p^r}^{*}$ and $y \in \mathbb{F}_{p^t}^{*}$, since $P^{T}$ is right $\mathbb{F}_{p^r}$-linear, we have $1=z_{c, y} \star \lambda_{z_{c, y}}^{-1}=z_{c, y} \star c^{-u}\pi(y)=c^{-u}(z_{c, y} \star \pi(y))$, and then $c^{u}=z_{c, y} \star \pi(y)=\rho^{-1}(c^{-1}y) \star \pi(y)$, which implies that $c=\rho^{-1}(c^{-v}y) \star \pi(y)$ for any $c \in \mathbb{F}_{p^r}^{*}$. Therefore, for any $x \in \mathbb{F}_{p^m}$ and $y \in \mathbb{F}_{p^t}^{*}$, we have 
	\begin{equation}\label{Th0 1}
		\begin{split}
		\mathrm{Tr}_{1}^{r}(c\phi(x,y))&=\mathrm{Tr}_{1}^{m}(c \psi(x, y))=\mathrm{Tr}_{1}^{m}(\psi(x, y)(\rho^{-1}(c^{-v}y) \star \pi(y)))\\
		&=\mathrm{Tr}_{1}^{m}(\rho^{-1}(c^{-v}y)(\psi(x, y) \circ \pi(y)))
		=\mathrm{Tr}_{1}^{m}(\rho^{-1}(c^{-v}y)x).
		\end{split}
	\end{equation}
	By $\psi(x, 0)=0, \rho(0)=0$, Eq. \eqref{Th0 1} also holds for $y=0$. Thus, for any $c \in \mathbb{F}_{p^r}^{*}$ and $y \in \mathbb{F}_{p^t}$, we have $\pi^{[c]}(y)=\rho^{-1}(c^{-v}y)$. Since $\rho$ is a permutation of $\mathbb{F}_{p^m}$, then $\pi^{[c]}$ is an injection from $\mathbb{F}_{p^t}$ to $\mathbb{F}_{p^m}$. For any $c \in \mathbb{F}_{p^r}^{*}$, $c^{-v}\mathbb{F}_{p^t}=\mathbb{F}_{p^t}$, and $\pi^{[c]}(\mathbb{F}_{p^t})=\rho^{-1}(\mathbb{F}_{p^t})$ for all $c \in \mathbb{F}_{p^r}^{*}$. Denote $S_{1}=\rho^{-1}(\mathbb{F}_{p^t})$, which contains $0$ as $\rho(0)=0$. For any $x \in S_{1}$, it is easy to see that $(\pi^{[c]})^{-1}(x)=c^{v}\rho(x)$. Note that as $(\pi^{[c]})^{-1}(S_{1})= \mathbb{F}_{p^t}$, we have $\rho(x) \in \mathbb{F}_{p^t}$ for any $x \in S_{1}$. Since $v =p^{r-j_{0}}$, then $(\pi^{[d]})^{-1}(x)+(\pi^{[e]})^{-1}(x)=(\pi^{[d+e]})^{-1}(x)$ for any $d \neq -e \in \mathbb{F}_{p^r}^{*}$. For any $y \in \mathbb{F}_{p^t}^{*}$, denote $y=\gamma_{y}\beta_{i_{y}}$, where $\gamma_{y} \in \mathbb{F}_{p^r}^{*}$ and $1 \leq i_{y} \leq \frac{p^r-1}{p^t-1}$. Then for $x \in S_{1}^{*}$, $\mathrm{Tr}_{1}^{r}(cH((\pi^{[c]})^{-1}(x)))=\mathrm{Tr}_{1}^{r}(cH(c^{v}\rho(x)))=\mathrm{Tr}_{1}^{r}(cH(c^{v}\gamma_{\rho(x)}\beta_{i_{\rho(x)}}))=\mathrm{Tr}_{1}^{r}(\gamma_{\rho(x)}^{-u}H(\beta_{i_{\rho(x)}}))$, which is independent of $c$. For $x=0$, $\mathrm{Tr}_{1}^{r}(cH((\pi^{[c]})^{-1}(x)))=\mathrm{Tr}_{1}^{r}(cH(c^{v}\rho(x)))=0$. By Lemma \ref{new1}, $\{D_{F, i}, i \in \mathbb{F}_{p^r}\}$ is an $(m-t)$-plateaued partition. 
\end{proof}

When the presemifield considered is the finite field, we further have the following theorem, from which all generated $p$-ary plateaued functions have no nonzero linear structure.

\begin{theorem} \label{Theorem 1}
Let $r, t, m$ be positive integers with $p^r>3, r \mid t, r \mid m, r \neq t, m, t<m, \frac{m}{r} \leq \frac{p^t-1}{p^r-1}$. Let $\mathbb{F}_{p^t}^{*}/\mathbb{F}_{p^r}^{*}=\{\bar{\beta}_{1}, \dots, \bar{\beta}_{\frac{p^t-1}{p^r-1}}\}$, and let $\bar{\alpha}_{i} \in \mathbb{F}_{p^m}^{*}/\mathbb{F}_{p^r}^{*}, 1 \leq i \leq \frac{p^t-1}{p^r-1}$, and $\alpha_{1}, \dots, \alpha_{\frac{m}{r}}$ be linearly independent over $\mathbb{F}_{p^r}$. Define $\pi: \mathbb{F}_{p^t} \rightarrow \mathbb{F}_{p^m}$ by 
$\pi(0)=0, \pi(c\beta_{i})=\theta(c)\alpha_{\tau(i)}, 1 \leq i \leq \frac{p^t-1}{p^r-1}, c \in \mathbb{F}_{p^r}^{*}$, where $\theta$ is a permutation of $\mathbb{F}_{p^r}^{*}$ satisfying $\theta^{-1}(d^{-1})+\theta^{-1}(e^{-1})=\theta^{-1}((d+e)^{-1})$ for any $d \neq -e \in \mathbb{F}_{p^r}^{*}$, and $\tau$ is a permutation of $\{i, 1 \leq i \leq \frac{p^t-1}{p^r-1}\}$. Let $H: \mathbb{F}_{p^t} \rightarrow \mathbb{F}_{p^r}$ satisfy $H(0)=0$ and $H(c\beta_{i})=\theta(c)H(\beta_{i})$ for any $1 \leq i \leq \frac{p^t-1}{p^r-1}, c \in \mathbb{F}_{p^r}^{*}$. Then $\{D_{F, i}, i \in \mathbb{F}_{p^r}\}$ is an $(m-t)$-plateaued partition for which any generated $p$-ary plateaued function has no nonzero linear structure, where 
$F: \mathbb{F}_{p^m} \times \mathbb{F}_{p^t} \rightarrow \mathbb{F}_{p^r}$ is given by
\begin{equation}\label{Th1 0}
	F(x, y)=\mathrm{Tr}_{r}^{m}(x \pi(y))+H(y), x \in \mathbb{F}_{p^m}, y \in \mathbb{F}_{p^t}.
\end{equation}
\end{theorem}

\begin{proof}
With the same notations as in Lemma \ref{new1}, $\pi^{[c]}(y)=c\pi(y), y \in \mathbb{F}_{p^t}$. By the definition of $\pi$, for every $c \in \mathbb{F}_{p^r}^{*}$, $c\pi$ is injective and $\{c\pi(y), y \in \mathbb{F}_{p^t}\}=\{0, \gamma\alpha_{i}, \gamma \in \mathbb{F}_{p^r}^{*}, 1 \leq i \leq \frac{p^t-1}{p^r-1}\}$. We denote this set by $S_{1}$. For any $x \in S_{1}^{*}$, denote $x=\gamma_{x} \alpha_{i_{x}}$, where $\gamma_{x} \in \mathbb{F}_{p^r}^{*}, 1 \leq i_{x} \leq \frac{p^t-1}{p^r-1}$. From the definition of $\pi$ and the properties of $\theta$ and $H$, for any $x \in S_{1}^{*}$ and $d \neq -e \in \mathbb{F}_{p^r}^{*}$, 
\begin{equation*}
	\begin{split}
	&(\pi^{[d]})^{-1}(x)+(\pi^{[e]})^{-1}(x)-(\pi^{[d+e]})^{-1}(x)=\pi^{-1}(d^{-1}x)+\pi^{-1}(e^{-1}x)-\pi^{-1}((d+e)^{-1}x)\\
	&=\pi^{-1}(d^{-1}\gamma_{x}\alpha_{i_{x}})+\pi^{-1}(e^{-1}\gamma_{x}\alpha_{i_{x}})-\pi^{-1}((d+e)^{-1}\gamma_{x}\alpha_{i_{x}})\\
	&=\theta^{-1}(d^{-1}\gamma_{x})\beta_{\tau^{-1}(i_{x})}+\theta^{-1}(e^{-1}\gamma_{x})\beta_{\tau^{-1}(i_{x})}-\theta^{-1}((d+e)^{-1}\gamma_{x})\beta_{\tau^{-1}(i_{x})}=0,
	\end{split}
\end{equation*}
and
\begin{equation*}
\begin{split}
\mathrm{Tr}_{1}^{r}(cH((\pi^{[c]})^{-1}(x)))&=\mathrm{Tr}_{1}^{r}(cH((\pi^{-1}(c^{-1}x)))=\mathrm{Tr}_{1}^{r}(cH(\pi^{-1}(c^{-1}\gamma_{x}\alpha_{i_{x}})))\\
&=\mathrm{Tr}_{1}^{r}(cH(\theta^{-1}(c^{-1}\gamma_{x})\beta_{\tau^{-1}(i_{x})}))=\mathrm{Tr}_{1}^{r}(\gamma_{x}H(\beta_{\tau^{-1}(i_{x})})).
\end{split}
\end{equation*}
Then by Lemma \ref{new1}, $\{D_{F, i}, i \in \mathbb{F}_{p^r}\}$ is an $(m-t)$-plateaued partition. Since $\{\alpha_{1}, \dots, \alpha_{\frac{m}{r}}\}$ is a basis of $\mathbb{F}_{p^m}$ over $\mathbb{F}_{p^r}$, then $S_{1} \times \mathbb{F}_{p^t}$ contains a basis of $\mathbb{F}_{p^m} \times \mathbb{F}_{p^t}$ over $\mathbb{F}_{p}$. Therefore any generated $p$-ary $(m-t)$-plateaued function has no nonzero linear structure by the proof of Lemma \ref{new1}, Proposition \ref{Pronew} and Remark \ref{Remark 1}. This completes the proof.
\end{proof}

\begin{remark}
Let $P(x)=\sum_{i=0}^{r-1}a_{i}x^{p^{i}}$ be a $p$-polynomial over $\mathbb{F}_{p^r}$ inducing a permutation of $\mathbb{F}_{p^r}$. Define $\theta(x)=\frac{1}{P^{-1}(x)}, x \in \mathbb{F}_{p^r}^{*}$. Then $\theta$ satisfies the condition in Theorem \ref{Theorem 1}.
\end{remark}

Based on Theorem \ref{Theorem 1} and the well-known direct sum construction, we obtain the following construction of $s$-plateaued partitions for which any generated $p$-ary $s$-plateaued function has no nonzero linear structure.

\begin{theorem} \label{Theorem 2}
Keep the same notations as in Theorem \ref{Theorem 1}. Let $F: \mathbb{F}_{p^m} \times \mathbb{F}_{p^t} \rightarrow \mathbb{F}_{p^r}$ be defined by Eq. \eqref{Th1 0}. Let $R: V_{k}^{(p)} \rightarrow \mathbb{F}_{p^r}$ be a vectorial dual-bent function satisfying Condition A. Define 
\begin{equation*}
	T(x, y, z)=P(F(x, y))+R(z), x \in \mathbb{F}_{p^m}, y \in \mathbb{F}_{p^t}, z \in V_{k}^{(p)},
\end{equation*}
where $P$ is a permutation of $\mathbb{F}_{p^r}$. Then $\{D_{T, i}, i \in \mathbb{F}_{p^r}\}$ is an $(m-t)$-plateaued partition for which any generated $p$-ary $(m-t)$-plateaued function has no nonzero linear structure.
\end{theorem}

\begin{proof}
For any $c \in \mathbb{F}_{p^r}^{*}$, $T_{c}(x, y, z)=P(F(x, y))_{c}+R_{c}(z), x \in \mathbb{F}_{p^m}, y \in \mathbb{F}_{p^t}, z \in V_{k}^{(p)}$, and then for any $a=(a_{1}, a_{2}, a_{3}) \in \mathbb{F}_{p^m} \times \mathbb{F}_{p^t} \times V_{k}^{(p)}$, 
\begin{equation*}
	W_{T_{c}}(a)=W_{P(F)_{c}}(a_{1}, a_{2})W_{R_{c}}(a_{3}).
\end{equation*}
Since $R$ is a vectorial dual-bent function satisfying Condition A, then $(R_{c})^{*}(a_{3})=(R^{*})_{c}(a_{3}), \epsilon_{R_{c}}=\epsilon$, where $\epsilon \in \{\pm 1\}$ is a constant. Combining Theorem \ref{Theorem 1}, Lemma \ref{new1}, Proposition \ref{Proposition 2}  (together with their proofs) and Remark \ref{Remark 1}, we have that $P(F(x, y))_{c}$ is $(m-t)$-plateaued for which $\epsilon_{P(F)_{c}}(a_{1}, a_{2})=1, a_{1}\in S_{1}, a_{2} \in \mathbb{F}_{p^t}$, and there exist functions $G: S_{1} \times \mathbb{F}_{p^t} \rightarrow \mathbb{F}_{p^r}$ and $g: S_{1} \times \mathbb{F}_{p^t} \rightarrow \mathbb{F}_{p}$ such that
\begin{equation*}
	(P(F(x, y))_{c})^{*}(a_{1}, a_{2})=\mathrm{Tr}_{1}^{r}(cP(G(a_{1}, a_{2})))+g(a_{1}, a_{2}), a_{1} \in S_{1}, a_{2} \in \mathbb{F}_{p^t},
\end{equation*}
where $S_{1}=\{0, \gamma \alpha_{i}, \gamma \in \mathbb{F}_{p^r}^{*}, 1 \leq i \leq \frac{p^t-1}{p^r-1}\}$ contains a basis of $\mathbb{F}_{p^m}$ over $\mathbb{F}_{p}$. Hence, for any $c \in \mathbb{F}_{p^r}^{*}$, $T_{c}$ is an $(m-t)$-plateaued function with $S_{T_{c}}=S_{1} \times \mathbb{F}_{p^t} \times V_{k}^{(p)}$ (denoted by $S$) and for any $a=(a_{1}, a_{2}, a_{3}) \in S$, $\epsilon_{T_{c}}(a)=\epsilon$, and
\begin{equation*}
	(T_{c})^{*}(a)=\mathrm{Tr}_{1}^{r}(cP(G(a_{1}, a_{2})))+g(a_{1}, a_{2})+\mathrm{Tr}_{1}^{r}(cR^{*}(a_{3})).
\end{equation*}
Denote $G'(a)=P(G(a_{1}, a_{2}))+R^{*}(a_{3})$ and $g'(a)=g(a_{1}, a_{2})$. Then $(T_{c})^{*}(a)=G'_{c}(a)+g'(a)$ for any $a \in S$. By Proposition \ref{Proposition 2}, $\{D_{T, i}, i \in \mathbb{F}_{p^r}\}$ is an $(m-t)$-plateaued partition. Since $S$ contains a basis of $\mathbb{F}_{p^m} \times \mathbb{F}_{p^t} \times V_{k}^{(p)}$ over $\mathbb{F}_{p}$, then any generated $p$-ary $(m-t)$-plateaued function has no nonzero linear structure by Proposition \ref{Pronew} and Remark \ref{Remark 1}. 
\end{proof}

By Theorem \ref{Theorem 1} and  \cite[Theorem 4]{WFW2023Be}, 
we obtain the following explicit construction of $s$-plateaued partitions for which any generated $p$-ary $s$-plateaued function has no nonzero linear structure.

\begin{theorem} \label{Theorem 3}
Let $r, t, m, k$ be positive integers with $p^r>3, r \mid t, r \mid m, r \mid k, r \neq t,k,m$, $t<m$ and $\frac{m}{r} \leq \frac{p^t-1}{p^r-1}$. Let $\mathbb{F}_{p^t}^{*}/\mathbb{F}_{p^r}^{*}=\{\bar{\beta}_{1}, \dots, \bar{\beta}_{\frac{p^t-1}{p^r-1}}\}$, and let $\bar{\alpha}_{i} \in \mathbb{F}_{p^m}^{*}/\mathbb{F}_{p^r}^{*}, 1 \leq i \leq \frac{p^t-1}{p^r-1}$, and $\alpha_{1}, \dots, \alpha_{\frac{m}{r}}$ be linearly independent over $\mathbb{F}_{p^r}$. For any $j \in \mathbb{F}_{p^r}$, let $P^{(j)}$ be a permutation of $\mathbb{F}_{p^r}$, $\pi^{(j)}: \mathbb{F}_{p^t} \rightarrow \mathbb{F}_{p^m}$ be given by $\pi^{(j)}(0)=0, \pi^{(j)}(c \beta_{i})=\theta^{(j)}(c)\alpha_{\tau^{(j)}(i)}, 1 \leq i \leq \frac{p^t-1}{p^r-1}, c \in \mathbb{F}_{p^r}^{*}$, where 
 $\theta^{(j)}$ is a permutation of $\mathbb{F}_{p^r}^{*}$ satisfying $(\theta^{(j)})^{-1}(d^{-1})+(\theta^{(j)})^{-1}(e^{-1})=(\theta^{(j)})^{-1}((d+e)^{-1})$ for any $d \neq -e \in \mathbb{F}_{p^r}^{*}$, and $\tau^{(j)}$ is a permutation of $\{i, 1 \leq i \leq \frac{p^t-1}{p^r-1}\}$, and let $H^{(j)}: \mathbb{F}_{p^t} \rightarrow \mathbb{F}_{p^r}$ satisfy $H^{(j)}(c\beta_{i})=\theta^{(j)}(c)H^{(j)}(\beta_{i})$ for any $1 \leq i \leq \frac{p^t-1}{p^r-1}, c \in \mathbb{F}_{p^r}^{*}$. Define $T: \mathbb{F}_{p^m} \times \mathbb{F}_{p^t} \times \mathbb{F}_{p^k} \times \mathbb{F}_{p^k} \rightarrow \mathbb{F}_{p^r}$ as
 \begin{equation} \label{Th3 0}
 	T(x, y, z, w)=F(\mathrm{Tr}_{r}^{k}(u Q(zw^{-1})); x, y)+\mathrm{Tr}_{r}^{k}(v Q(zw^{-1})),
 \end{equation}
 where $u, v \in \mathbb{F}_{p^k}$ are linearly independent over $\mathbb{F}_{p^r}$, $Q$ is a permutation of $\mathbb{F}_{p^k}$, and for each $j \in \mathbb{F}_{p^r}$, $F(j; x, y): \mathbb{F}_{p^m} \times \mathbb{F}_{p^t} \rightarrow \mathbb{F}_{p^r}$ is given by
 \begin{equation*}
 	F(j; x, y)=P^{(j)}(\mathrm{Tr}_{r}^{m}(x \pi^{(j)}(y))+H^{(j)}(y)), x \in \mathbb{F}_{p^m}, y \in \mathbb{F}_{p^t}.
 \end{equation*}
 Then $\{D_{T, i}, i \in \mathbb{F}_{p^r}\}$ is an $(m-t)$-plateaued partition for which any generated $p$-ary $(m-t)$-plateaued function has no nonzero linear structure.
\end{theorem}
\begin{proof}
For any $c \in \mathbb{F}_{p^r}^{*}$ and $a=(a_{1}, a_{2}, a_{3}, a_{4}) \in \mathbb{F}_{p^m} \times \mathbb{F}_{p^t} \times \mathbb{F}_{p^k} \times \mathbb{F}_{p^k}$, by the proof of \cite[Theorem 4]{WFW2023Be},
\begin{equation*}
	W_{T_{c}}(a)=p^{k}\zeta_{p}^{\mathrm{Tr}_{1}^{k}(cv Q(-a_{3}^{-1}a_{4}))}W_{(F(\mathrm{Tr}_{r}^{k}(u Q(-a_{3}^{-1}a_{4})); x, y))_{c}}(a_{1}, a_{2}).
\end{equation*}
Combining Theorem \ref{Theorem 1}, Lemma \ref{new1}, Proposition \ref{Proposition 2} (together with their proofs) and Remark \ref{Remark 1}, $F(j; x, y) \ (j \in \mathbb{F}_{p^r})$ is a vectorial $(m-t)$-plateaued function for which for any $c \in \mathbb{F}_{p^r}^{*}$, $S_{F(j; x, y)_{c}}=S_{1} \times \mathbb{F}_{p^t}$
and there exist functions $G^{(j)}: S_{1} \times \mathbb{F}_{p^t} \rightarrow \mathbb{F}_{p^r}$ and $g^{(j)}: S_{1} \times \mathbb{F}_{p^t} \rightarrow \mathbb{F}_{p}$ such that
\begin{equation*}
	(F(j; x, y)_{c})^{*}(a_{1}, a_{2})=\mathrm{Tr}_{1}^{r}(cP^{(j)}(G^{(j)}(a_{1}, a_{2})))+g^{(j)}(a_{1}, a_{2}), a_{1} \in S_{1}, a_{2} \in \mathbb{F}_{p^t},
\end{equation*}
where $S_{1}=\{0, \gamma \alpha_{i}, \gamma \in \mathbb{F}_{p^r}^{*}, 1 \leq i \leq \frac{p^t-1}{p^r-1}\}$ contains a basis of $\mathbb{F}_{p^m}$ over $\mathbb{F}_{p}$. Therefore, for any $c \in \mathbb{F}_{p^r}^{*}$, $T_{c}$ is an $(m-t)$-plateaued function with $S_{T_{c}}=S_{1} \times \mathbb{F}_{p^t} \times \mathbb{F}_{p^k} \times \mathbb{F}_{p^k}$ (denoted by $S$) and for any $a=(a_{1}, a_{2}, a_{3}, a_{4}) \in S$,
\begin{equation*}
	\begin{split}
	(T_{c})^{*}(a)&=\mathrm{Tr}_{1}^{r}\Big(cP^{(\mathrm{Tr}_{r}^{k}(u Q(-a_{3}^{-1}a_{4})))}\big(G^{(\mathrm{Tr}_{r}^{k}(u Q(-a_{3}^{-1}a_{4})))}(a_{1}, a_{2})\big)\Big)+g^{(\mathrm{Tr}_{r}^{k}(u Q(-a_{3}^{-1}a_{4})))}(a_{1}, a_{2})\\
	& \ \ \ +\mathrm{Tr}_{1}^{k}(cv Q(-a_{3}^{-1}a_{4})).
	\end{split}
\end{equation*}
Denote 
\begin{equation*}
	G(a)=P^{(\mathrm{Tr}_{r}^{k}(u Q(-a_{3}^{-1}a_{4})))}\big(G^{(\mathrm{Tr}_{r}^{k}(u Q(-a_{3}^{-1}a_{4})))}(a_{1}, a_{2})\big)+\mathrm{Tr}_{r}^{k}(vQ(-a_{3}^{-1}a_{4}))
\end{equation*} and 
\begin{equation*}
	g(a)=g^{(\mathrm{Tr}_{r}^{k}(u Q(-a_{3}^{-1}a_{4})))}(a_{1}, a_{2}).
\end{equation*}
Then $(T_{c})^{*}(a)=G_{c}(a)+g(a)$ for any $c \in \mathbb{F}_{p^r}^{*}$ and $a=(a_{1}, a_{2}, a_{3}, a_{4}) \in S$. By Proposition \ref{Proposition 2}, $\{D_{T, i}, i \in \mathbb{F}_{p^r}\}$ is an $(m-t)$-plateaued partition. Since $S$ contains a basis of $\mathbb{F}_{p^m} \times \mathbb{F}_{p^t} \times \mathbb{F}_{p^k} \times \mathbb{F}_{p^k}$ over $\mathbb{F}_{p}$, then any generated $p$-ary $(m-t)$-plateaued function has no nonzero linear structure by Proposition \ref{Pronew} and Remark \ref{Remark 1}. 
\end{proof}

\begin{remark}
Let $H$ in Theorems \ref{Theorem 0}-\ref{Theorem 2} (resp., $H^{(j)}$ in Theorem \ref{Theorem 3}) be the zero function, then the corresponding plateaued partitions $\Gamma=\{A_{i}, i \in \mathbb{F}_{p^r}\}$ all satisfy $-A_{i}=A_{i}$ for any $i \in \mathbb{F}_{p^r}$.
\end{remark}

\section{Characterizations of $s$-plateaued partitions with $A_{i}=-A_{i}$}
In this section, we characterize $s$-plateaued partitions $\Gamma=\{A_{i}, 1 \leq i \leq K\}$, where $p$ is odd, $K \geq 5$ and $-A_{i}=A_{i}, 1 \leq i \leq K$. In particular, when $s=0$, we partially address an open problem on bent partitions. First, we need two lemmas.

\begin{lemma} \label{Lemma 1}
	Let $k, b$ be positive integers, $p$ be an odd prime. Let $L=\mathbb{Q}(\zeta_{p})$ and  $\mathcal{O}_{L}=\mathbb{Z}[\zeta_{p}]$ be the ring of algebraic integers in $L$. If there exists a real number $a \in \mathcal{O}_{L}$ such that $a\mathcal{O}_{L}=b(1-\zeta_{p})^{k}\mathcal{O}_{L}$, then $k$ is even.
\end{lemma}

\begin{proof}
Let $\sigma_{-1}$ be the automorphism of $L$ defined by $\sigma_{-1}(\zeta_{p})=\zeta_{p}^{-1}$. Then for any $\alpha \in L$, $\sigma_{-1}(\alpha)=\overline{\alpha}$, and $\alpha$ is a real number if and only if $\sigma_{-1}(\alpha)=\alpha$.
From $a\mathcal{O}_{L}=b(1-\zeta_{p})^{k}\mathcal{O}_{L}$, we have that $a=b(1-\zeta_{p})^{k}u$, where $u$ is a unit of $\mathcal{O}_{L}$. It is a well-known result that every unit of $\mathcal{O}_{L}$ can be written as the product of a real unit and a root of unity \cite{Feng2018Al,IR1990A}. Namely, $u=r\zeta_{p}^{i}$ for some real unit $r$ and $0 \leq i \leq p-1$. Then $(1-\zeta_{p})^{k}\zeta_{p}^{i}=\frac{a}{br}$ is a real number. Thus, $(1-\zeta_{p})^{k}\zeta_{p}^{i}=\sigma_{-1}((1-\zeta_{p})^{k}\zeta_{p}^{i})=(1-\zeta_{p}^{-1})^{k}\zeta_{p}^{-i}=(-1)^{k}(1-\zeta_{p})^{k}\zeta_{p}^{-k-i}$, which implies that $k$ is even (otherwise, $\zeta_{p}^{2i+k}=(-1)^{k}=-1$ which is impossible).
\end{proof}

The following lemma follows from the proof of \cite[Theorem 2]{WWF2026Fu}.
\begin{lemma}[\cite{WWF2026Fu}]\label{Lemma 2}
	Let $p\geq 5$ be an odd prime, $n+s$ be even, and $a$ be a fixed element in $V_{n}^{(p)}$. Assume that $f: V_{n}^{(p)} \rightarrow  \mathbb{F}_{p}$ satisfies that  $W_{P(f)}(a)=\epsilon_{a} p^{\frac{n+s}{2}}\zeta_{p}^{\gamma_{_{P(f), a}}}$ for any permutation $P$ of $\mathbb{F}_{p}$, where $\epsilon_{a} \in \{\pm 1\}$ is independent of $P$ and $\gamma_{_{P(f), a}} \in \mathbb{F}_{p}$. Define $M_{j}^{(a)}=\frac{1-\zeta_{p}^{\gamma_{_{P_{0,j}(f), a}}-\gamma_{_{f, a}}}}{1-\zeta_{p}^{j}}, j \in \mathbb{F}_{p}^{*}$, where $P_{0, j}=(0 j)$ is the transposition permutation. Then either $M_{j}^{(a)}=1, j \in \mathbb{F}_{p}^{*}$, or there is a unique $j_{a} \in \mathbb{F}_{p}^{*}$ such that $M_{j_{a}}^{(a)}=-\zeta_{p}^{-j_{a}}, M_{j}^{(a)}=0, j \in \mathbb{F}_{p}^{*} \backslash \{j_{a}\}$.
\end{lemma}

For an $s$-plateaued partition $\Gamma=\{A_{i}, 1 \leq i \leq K\}$ with $-A_{i}=A_{i}$, where $p$ is odd and $K \geq 5$, the following theorem shows that the conditions in Proposition \ref{Proposition 2} are also necessary, where the corresponding function $g=0$.

\begin{theorem}\label{Theorem 4}
	Let $p$ be an odd prime and $K \geq 5$ be an integer divisible by $p$. Let $\Gamma=\{A_{i}, 1 \leq i \leq K\}$ be a partition of $V_{n}^{(p)}$, where $-A_{i}=A_{i}, 1 \leq i \leq K$. Then the following statements are equivalent.
	
	(i) $\Gamma$ is an $s$-plateaued partition.
	
	(ii) $K$ is a power of $p$ (denoted by $K=p^m$), and if  $\Gamma$ is written as $\Gamma=\{A_{i}, i \in V_{m}^{(p)}\}$, then $F: V_{n}^{(p)} \rightarrow V_{m}^{(p)}$ defined as 
	\begin{equation*}
		F(x)=i, x \in A_{i}, i \in V_{m}^{(p)}
	\end{equation*}
	is a vectorial $s$-plateaued function satisfying the following conditions:
	\begin{itemize}
		\item[(1)] The Walsh supports of all component functions $S_{F_{c}}, c \in {V_{m}^{(p)}}^{*}$, are all the same (denoted by $S$).
		\item[(2)] Each component function $F_{c}$ ($c \in {V_{m}^{(p)}}^{*}$) satisfies $\epsilon_{F_{c}}(a)=\epsilon(a)$ for any $a \in S$, where $\epsilon(a)$ is independent of $c$. 
		\item[(3)] There is a function $G: S \rightarrow V_{m}^{(p)}$ such that $(F_{c})^{*}(x)=G_{c}(x), x \in S, c \in {V_{m}^{(p)}}^{*}$.
	\end{itemize}
\end{theorem}

\begin{proof}
	By Proposition \ref{Proposition 2}, $(ii) \Rightarrow (i)$ holds. In the following, we prove $(i) \Rightarrow (ii)$. As $-A_{i}=A_{i}, 1 \leq i \leq K$, all $p$-ary $s$-plateaued functions $f$ generated by $\Gamma$ satisfy $f(x)=f(-x)$, and then $|D_{f, f(0)}|$ is odd and $|D_{f, j}| \ (j \neq f(0))$ is even. Thus $\Gamma$ is unbalanced. By Proposition \ref{Proposition 1}, $n+s$ is even. Define $\widetilde{F}: V_{n}^{(p)} \rightarrow \{1, 2, \dots, K\}$ as $\widetilde{F}(x)=i, x \in A_{i}, 1 \leq i \leq K$. By the definition of $s$-plateaued partitions, $B(\widetilde{F}(x))$ is a $p$-ary $s$-plateaued function for any balanced map from $\{1, 2, \dots, K\}$ to $\mathbb{F}_{p}$.
	We first show that the Walsh supports of all $p$-ary $s$-plateaued functions generated by $\Gamma$ are the same. Let $a$ be an arbitrary fixed element in $V_{n}^{(p)}$. With the same arguments as in the proof of Proposition \ref{Proposition 1} $(i)$, we only need to prove that $\delta_{S_{B(P(\widetilde{F}(x)))}}(a)=\delta_{S_{B(\widetilde{F}(x))}}(a)$ for any balanced map $B$ from $\{1, 2, \dots, K\}$ to $\mathbb{F}_{p}$ and any transposition permutation $P$ of $\{1, 2, \dots, K\}$. Let $B$ be an arbitrary fixed balanced map from $\{1, 2, \dots, K\}$ to $\mathbb{F}_{p}$ and $P=(i_{0}i_{1})$ be an arbitrary fixed transposition permutation. We use the proof by contradiction. W.l.o.g., assume that $a \in S_{B(P(\widetilde{F}))}, a \notin S_{B(\widetilde{F})}$. We have 
	\begin{equation} \label{Th4 1}
		\begin{split}
			\epsilon_{B(P(\widetilde{F}))}(a)p^{\frac{n+s}{2}}\zeta_{p}^{(B(P(\widetilde{F})))^{*}(a)}
			&=W_{B(P(\widetilde{F}))}(a)=\sum_{x \in V_{n}^{(p)}}\zeta_{p}^{B(P(\widetilde{F}(x)))-\langle a, x\rangle_{n}}=\sum_{i=1}^{K}\zeta_{p}^{B(P(i))}\chi_{-a}(A_{i})\\
			&=\sum_{i \in \{1, \dots, K\} \backslash \{i_{0}, i_{1}\}}\zeta_{p}^{B(i)}\chi_{-a}(A_{i})
			+\zeta_{p}^{B(i_{1})}\chi_{-a}(A_{i_{0}})+\zeta_{p}^{B(i_{0})}\chi_{-a}(A_{i_{1}}),
		\end{split}
	\end{equation}
	and
	\begin{equation} \label{Th4 2}
		\begin{split}
			0=W_{B(\widetilde{F})}(a)
			&=\sum_{x \in V_{n}^{(p)}}\zeta_{p}^{B(\widetilde{F}(x))-\langle a, x\rangle_{n}}=\sum_{i=1}^{K}\zeta_{p}^{B(i)}\chi_{-a}(A_{i})\\
			&=\sum_{i \in \{1, \dots, K\} \backslash \{i_{0}, i_{1}\}}\zeta_{p}^{B(i)}\chi_{-a}(A_{i})
			+\zeta_{p}^{B(i_{0})}\chi_{-a}(A_{i_{0}})+\zeta_{p}^{B(i_{1})}\chi_{-a}(A_{i_{1}}).
		\end{split}
	\end{equation}
	Subtracting Eq. \eqref{Th4 1} from Eq. \eqref{Th4 2}, we have
	\begin{equation}\label{Th4 3}
		\epsilon_{B(P(\widetilde{F}))}(a)p^{\frac{n+s}{2}}\zeta_{p}^{(B(P(\widetilde{F})))^{*}(a)}=(\zeta_{p}^{B(i_{0})}-\zeta_{p}^{B(i_{1})})(\chi_{-a}(A_{i_{1}})-\chi_{-a}(A_{i_{0}})).
	\end{equation}
	Case I: $B(i_{0})=B(i_{1})$ or $\chi_{-a}(A_{i_{0}})=\chi_{-a}(A_{i_{1}})$. This case contradicts Eq. \eqref{Th4 3}. \\
	Case II: $B(i_{0}) \neq B(i_{1})$ and $\chi_{-a}(A_{i_{0}}) \neq \chi_{-a}(A_{i_{1}})$. Let $L=\mathbb{Q}(\zeta_{p})$ and $\mathcal{O}_{L}=\mathbb{Z}[\zeta_{p}]$. Since $-A_{i}=A_{i}$, then $\sigma_{-1}(\chi_{-a}(A_{i}))=\sigma_{-1}(\sum_{x \in A_{i}}\zeta_{p}^{\langle -a, x\rangle_{n}})=\sum_{x \in A_{i}}\zeta_{p}^{\langle a, x\rangle_{n}}=\sum_{x \in A_{i}}\zeta_{p}^{\langle a, -x\rangle_{n}}=\chi_{-a}(A_{i})$, where $\sigma_{-1}$ is the automorphism of $L$ defined by $\sigma_{-1}(\zeta_{p})=\zeta_{p}^{-1}$. Thus, $\chi_{-a}(A_{i}) \in \mathcal{O}_{L}$ is a real number. Since $p\mathcal{O}_{L}=((1-\zeta_{p})\mathcal{O}_{L})^{p-1}$ and $(\zeta_{p}^{B(i_{0})}-\zeta_{p}^{B(i_{1})})\mathcal{O}_{L}=(1-\zeta_{p}^{B(i_{1})-B(i_{0})})\mathcal{O}_{L}=(1-\zeta_{p})\mathcal{O}_{L}$, by Eq. \eqref{Th4 3}, we have 
	\begin{equation}\label{Th4 4}
		((1-\zeta_{p})\mathcal{O}_{L})^{\frac{(p-1)(n+s)}{2}}=(1-\zeta_{p})\mathcal{O}_{L} \cdot (\chi_{-a}(A_{i_{1}})-\chi_{-a}(A_{i_{0}}))\mathcal{O}_{L}.
	\end{equation}
	Then $(\chi_{-a}(A_{i_{1}})-\chi_{-a}(A_{i_{0}}))\mathcal{O}_{L}=(1-\zeta_{p})^{\frac{(p-1)(n+s)}{2}-1}\mathcal{O}_{L}$, which is impossible by Lemma \ref{Lemma 1} since $\chi_{-a}(A_{i_{1}})-\chi_{-a}(A_{i_{0}})$ is a real number and $\frac{(p-1)(n+s)}{2}-1$ is odd. Therefore, the Walsh supports of all $p$-ary $s$-plateaued functions generated by $\Gamma$ are the same (denoted by $S$). Note that $0 \in S$ as $\Gamma$ is unbalanced.
	
	Let $a$ be an arbitrary fixed element in $S$. Now we show that for all $p$-ary $s$-plateaued functions $f$ generated by $\Gamma$, $\epsilon_{f}(a)$ are the same. With the similar arguments as in the proof of Proposition \ref{Proposition 1} $(i)$, we only need to prove that $\epsilon_{B(P(\widetilde{F}(x)))}(a)=\epsilon_{B(\widetilde{F}(x))}(a)$ for any balanced map $B$ from $\{1, 2, \dots, K\}$ to $\mathbb{F}_{p}$ and any  transposition permutation $P$ of $\{1, 2, \dots, K\}$. Let $B$ be an arbitrary fixed balanced map from $\{1, 2, \dots, K\}$ to $\mathbb{F}_{p}$ and $P=(i_{0}i_{1})$ be an arbitrary fixed transposition permutation. We use the proof by contradiction. W.l.o.g, assume that $\epsilon_{B(P(\widetilde{F}(x)))}(a)=1, \epsilon_{B(\widetilde{F}(x))}(a)=-1$. By Eq. \eqref{Th4 1}, we can obtain
	\begin{equation}\label{Th4 5}
		p^{\frac{n+s}{2}}(\zeta_{p}^{(B(P(\widetilde{F})))^{*}(a)}+\zeta_{p}^{(B(\widetilde{F}))^{*}(a)})
		=(\zeta_{p}^{B(i_{0})}-\zeta_{p}^{B(i_{1})})(\chi_{-a}(A_{i_{1}})-\chi_{-a}(A_{i_{0}})).
	\end{equation}
	Case I: $B(i_{0})=B(i_{1})$ or $\chi_{-a}(A_{i_{0}})=\chi_{-a}(A_{i_{1}})$. This case contradicts Eq. \eqref{Th4 5}. \\
	Case II: $B(i_{0}) \neq B(i_{1})$, $\chi_{-a}(A_{i_{0}}) \neq \chi_{-a}(A_{i_{1}})$ and $(B(P(\widetilde{F})))^{*}(a)=(B(\widetilde{F}))^{*}(a)$. Since $p \mathcal{O}_{L}=((1-\zeta_{p})\mathcal{O}_{L})^{p-1}$ and $(\zeta_{p}^{B(i_{0})}-\zeta_{p}^{B(i_{1})}) \mathcal{O}_{L}=(1-\zeta_{p})\mathcal{O}_{L}$, then we have $(\chi_{-a}(A_{i_{1}})-\chi_{-a}(A_{i_{0}}))\mathcal{O}_{L}=2(1-\zeta_{p})^{\frac{(p-1)(n+s)}{2}-1}\mathcal{O}_{L}$, which is impossible by Lemma \ref{Lemma 1}.\\
	Case III: $B(i_{0}) \neq B(i_{1})$, $\chi_{-a}(A_{i_{0}}) \neq \chi_{-a}(A_{i_{1}})$ and $(B(P(\widetilde{F})))^{*}(a) \neq (B(\widetilde{F}))^{*}(a)$. Since $p \mathcal{O}_{L}=((1-\zeta_{p})\mathcal{O}_{L})^{p-1}$,  $(\zeta_{p}^{B(i_{0})}-\zeta_{p}^{B(i_{1})}) \mathcal{O}_{L}=(1-\zeta_{p})\mathcal{O}_{L}$ and $(\zeta_{p}^{(B(P(\widetilde{F})))^{*}(a)}+\zeta_{p}^{(B(\widetilde{F}))^{*}(a)})\mathcal{O}_{L}=\mathcal{O}_{L}$, then we have $(\chi_{-a}(A_{i_{1}})-\chi_{-a}(A_{i_{0}}))\mathcal{O}_{L}=(1-\zeta_{p})^{\frac{(p-1)(n+s)}{2}-1}\mathcal{O}_{L}$, which is impossible by Lemma \ref{Lemma 1}. Therefore, for all $p$-ary $s$-plateaued functions $f$ generated by $\Gamma$, $\epsilon_{f}(a)$ are the same (denoted by $\epsilon(a)$). 
	
	In the following, we compute the values $\chi_{-a}(A_{i}), 1 \leq i \leq K$, where $a$ is an arbitrary fixed element in $S$.
	Let $f$ be an arbitrary fixed $p$-ary $s$-plateaued function generated by $\Gamma$. Note that $cf \ (c \in \mathbb{F}_{p}^{*})$ is also a $p$-ary $s$-plateaued function generated by $\Gamma$. Then $S_{cf}=S_{f}=S$. By \cite[Theorem 1]{CM2013A}, for any $c \in \mathbb{F}_{p}^{*}$, we have $S_{cf}=cS_{f}$ (which yields that $\mathbb{F}_{p}^{*}S=S$) and $(cf)^{*}(a)=cf^{*}(c^{-1}a)$.
	Since $f(x)=f(-x)$, then $W_{f}(a)=W_{f}(-a)$, and thus  $f^{*}(a)=f^{*}(-a)$. 
	When $p=3$, then $(cf)^{*}(a)=cf^{*}(a), c \in \mathbb{F}_{p}^{*}$. When $p \geq 5$, we claim that $(cf)^{*}(a)=cf^{*}(a), c \in \mathbb{F}_{p}^{*}$ also holds. For any permutation $P$ of $\mathbb{F}_{p}$, $P(f)$ is also a $p$-ary $s$-plateaued function generated by $\Gamma$. Then $\epsilon_{P(f)}(a)=\epsilon(a)$. Define $M_{j}^{(a)}\triangleq \frac{1-\zeta_{p}^{(P_{0,j}(f))^{*}(a)-f^{*}(a)}}{1-\zeta_{p}^{j}}$, where $j \in \mathbb{F}_{p}^{*}$ and  $P_{0,j}=(0j)$ is the transposition permutation of $\mathbb{F}_{p}$. From $\epsilon(a)p^{\frac{n+s}{2}}\zeta_{p}^{(P_{0,j}(f))^{*}(a)}=W_{P_{0, j}(f)}(a)=\sum_{i \in \mathbb{F}_{p} \backslash \{0, j\}}\chi_{-a}(D_{f, i})\zeta_{p}^{i}+\chi_{-a}(D_{f, 0})\zeta_{p}^{j}+\chi_{-a}(D_{f, j})$ and $\epsilon(a)p^{\frac{n+s}{2}}\zeta_{p}^{f^{*}(a)}=W_{f}(a)=\sum_{i \in \mathbb{F}_{p} \backslash \{0, j\}}\chi_{-a}(D_{f, i})\zeta_{p}^{i}+\chi_{-a}(D_{f, 0})+\chi_{-a}(D_{f, j})\zeta_{p}^{j}$, we get $\chi_{-a}(D_{f, j})=\chi_{-a}(D_{f, 0})-\epsilon(a)p^{\frac{n+s}{2}}\zeta_{p}^{f^{*}(a)}M_{j}^{(a)}$ for any $j \in \mathbb{F}_{p}^{*}$. Then for any $c \in \mathbb{F}_{p}^{*}$, 
	\begin{equation*}
		\epsilon(a)p^{\frac{n+s}{2}}\zeta_{p}^{(cf)^{*}(a)}=W_{cf}(a)=\sum_{j \in \mathbb{F}_{p}}\chi_{-a}(D_{f, j})\zeta_{p}^{cj}=-\epsilon(a)p^{\frac{n+s}{2}}\zeta_{p}^{f^{*}(a)}\sum_{j \in \mathbb{F}_{p}^{*}}M_{j}^{(a)}\zeta_{p}^{cj}.
	\end{equation*}
	By Lemma \ref{Lemma 2}, $M_{j}^{(a)}=1, j \in \mathbb{F}_{p}^{*}$, or there is a unique $j_{a} \in \mathbb{F}_{p}^{*}$ such that $M_{j_{a}}^{(a)}=-\zeta_{p}^{-j_{a}}, M_{j}^{(a)}=0, j \in \mathbb{F}_{p}^{*} \backslash \{j_{a}\}$. Denote $E=\{a \in S: M_{j}^{(a)}=1, j \in \mathbb{F}_{p}^{*}\}$. Then when $a \in E$, $\sum_{j \in \mathbb{F}_{p}^{*}}M_{j}^{(a)}\zeta_{p}^{cj}=-1$, and then $(cf)^{*}(a)=f^{*}(a)$; when $a \in S \backslash E$, $\sum_{j \in \mathbb{F}_{p}^{*}}M_{j}^{(a)}\zeta_{p}^{cj}=-\zeta_{p}^{(c-1)j_{a}}$, and then $(cf)^{*}(a)=f^{*}(a)+(c-1)j_{a}$. Thus, $(cf)^{*}(a)=cg(a)+h(a)$, where $g: S \rightarrow \mathbb{F}_{p}$ is given as $g(a)=0, a \in E$, and $g(a)=j_{a}, a \in S \backslash E$, and $h: S \rightarrow \mathbb{F}_{p}$ is given as $h(a)=f^{*}(a), a \in E$, and $h(a)=f^{*}(a)-j_{a}, a \in S \backslash E$. Then for any $a \in S$,  $-g(a)+h(a)=(-f)^{*}(a)=-f^{*}(-a)=-f^{*}(a)=-g(a)-h(a)$, which implies that $h(a)=0$ and $g(a)=f^{*}(a)$. Thus, $(cf)^{*}(a)=cf^{*}(a), c \in \mathbb{F}_{p}^{*}$. Then by  \cite[Proposition 3]{WF2023Ne}, for any $j \in \mathbb{F}_{p}$,
	\begin{equation*}
		\begin{split}
			\chi_{-a}(D_{f, j})&=p^{n-1}\delta_{0}(a)+p^{-1}\sum_{c \in \mathbb{F}_{p}^{*}}W_{cf}(a)\zeta_{p}^{-cj}=p^{n-1}\delta_{0}(a)+\epsilon(a)p^{\frac{n+s}{2}-1}\sum_{c \in \mathbb{F}_{p}^{*}}\zeta_{p}^{c(f^{*}(a)-j)}\\
			&=p^{n-1}\delta_{0}(a)+\epsilon(a)p^{\frac{n+s}{2}-1}(p\delta_{j}(f^{*}(a))-1).
		\end{split}
	\end{equation*}
	For any fixed $a \in S$, since $\chi_{-a}(D_{f, j}), j \in \mathbb{F}_{p}$, take exactly two different values for any $s$-plateaued function $f$ generated by $\Gamma$, where the value $p^{n-1}\delta_{0}(a)+\epsilon(a)p^{\frac{n+s}{2}-1}(p-1)$ occurs one time and the value $p^{n-1}\delta_{0}(a)-\epsilon(a)p^{\frac{n+s}{2}-1}$ occurs $(p-1)$ times, we can obtain
	\begin{equation*}
		\chi_{-a}(A_{i})=\left\{
		\begin{split}
			\frac{p^n}{K}\delta_{0}(a)-\epsilon(a)\frac{p^{\frac{n+s}{2}}}{K}+\epsilon(a)p^{\frac{n+s}{2}}, & \text{ if } i=\widetilde{G}(a),\\
			\frac{p^n}{K}\delta_{0}(a)-\epsilon(a)\frac{p^{\frac{n+s}{2}}}{K}, &  \text{ if } i \in \{1, 2, \dots, K\} \backslash \{\widetilde{G}(a)\},
		\end{split}
		\right.
	\end{equation*}
	where $\widetilde{G}$ is some function from $S$ to $\{1, 2, \dots, K\}$. For $a \in S^{*}$, since $\chi_{-a}(A_{i}) \in \mathbb{Z}[\zeta_{p}] \cap \mathbb{Q}=\mathbb{Z}$, then $K$ is a divisor of $p^{\frac{n+s}{2}}$ and $K$ is a power of $p$ (denoted by $K=p^m$). Thus, we can view $\Gamma$ as $\Gamma=\{A_{i}, i \in V_{m}^{(p)}\}$ and (1) and (2) hold as $F_{c}$ is a $p$-ary $s$-plateaued function generated by $\Gamma$, and $\chi_{-a}(A_{i})=p^{n-m}\delta_{0}(a)-\epsilon(a)p^{\frac{n+s}{2}-m}+\delta_{i}(G(a))\epsilon(a)p^{\frac{n+s}{2}}, a \in S, i \in V_{m}^{(p)}$, where $G$ is some function from $S$ to $V_{m}^{(p)}$. Then for any $a \in S$, $W_{F_{c}}(a)=\sum_{i \in V_{m}^{(p)}}\chi_{-a}(A_{i})\zeta_{p}^{\langle c, i\rangle_{m}}=\epsilon(a)p^{\frac{n+s}{2}}\zeta_{p}^{G_{c}(a)}$, which shows that (3) holds. This completes the proof.
\end{proof}

\begin{remark}
	By the proof of Proposition \ref{Proposition 2}, the condition (3) in Theorem \ref{Theorem 4} can be replaced by the following condition:
	\begin{itemize}
		\item For any $d \neq -e \in {V_{m}^{(p)}}^{*}$ and $x \in S$, $(F_{d})^{*}(x)+(F_{e})^{*}(x)-(F_{d+e})^{*}(x)=0$.
	\end{itemize}
\end{remark}

When $K=p$, we obtain the following corollary.

\begin{corollary}\label{Corollary 1}
	Let $p\geq 5$ be an odd prime and $f: V_{n}^{(p)} \rightarrow \mathbb{F}_{p}$ with $f(x)=f(-x)$. Then the following statements are equivalent.
	
	(i) $\{D_{f, i}, i \in \mathbb{F}_{p}\}$ is an $s$-plateaued partition.
	
	(ii) $f$ is an $s$-plateaued function of $(p-1)$-form, where $n+s$ is even.
\end{corollary}

\begin{proof}
	$(i) \Rightarrow (ii)$: According to Theorem \ref{Theorem 4} and its proof, $f$ is an $s$-plateaued function, where $n+s$ is even, and for any $c \in \mathbb{F}_{p}^{*}$, $S_{f}=S_{cf}=cS_{f}$ and  $\epsilon_{cf}(a)=\epsilon_{f}(a), (cf)^{*}(a)=cf^{*}(a)$ for all $a \in S_{f}$. Then by \cite[Theorem 1]{CM2013A}, 
	for any $c \in \mathbb{F}_{p}^{*}$ and $a \in S_{f}$, $\epsilon_{f}(c^{-1}a)=\epsilon_{cf}(a)=\epsilon_{f}(a), cf^{*}(c^{-1}a)=(cf)^{*}(a)=cf^{*}(a)$, which implies that $\epsilon_{f}(c^{-1}a)=\epsilon_{f}(a), f^{*}(c^{-1}a)=f^{*}(a)$. Thus, for any $c \in \mathbb{F}_{p}^{*}$ and $x \in V_{n}^{(p)}$, we have
	\begin{equation*}
		\begin{split}
			p^{n}\zeta_{p}^{f(cx)}&=\sum_{a \in S_{f}}W_{f}(a)\zeta_{p}^{\langle a, cx\rangle_{n}}=\sum_{a \in S_{f}}W_{f}(c^{-1}a)\zeta_{p}^{\langle a,x\rangle_{n}}=\sum_{a \in S_{f}}\epsilon_{f}(c^{-1}a)p^{\frac{n+s}{2}}\zeta_{p}^{f^{*}(c^{-1}a)+\langle a, x\rangle_{n}}\\
			&=\sum_{a \in S_{f}}\epsilon_{f}(a)p^{\frac{n+s}{2}}\zeta_{p}^{f^{*}(a)+\langle a, x\rangle_{n}}=\sum_{a \in S_{f}}W_{f}(a)\zeta_{p}^{\langle a, x\rangle_{n}}=p^{n}\zeta_{p}^{f(x)},
		\end{split}
	\end{equation*}
	and hence $f(cx)=f(x)$, that is, $f$ is of $(p-1)$-form.
	
	$(ii) \Rightarrow (i)$: Since $f$ is an $s$-plateaued function of $(p-1)$-form, where $n+s$ is even, then for any $a \in S_{f}$ and $c \in \mathbb{F}_{p}^{*}$, we have
	\begin{equation*}
		W_{f}(a)=\sum_{x \in V_{n}^{(p)}}\zeta_{p}^{f(x)-\langle a, x\rangle_{n}}=\sum_{x \in V_{n}^{(p)}}\zeta_{p}^{f(cx)-\langle a, cx\rangle_{n}}=\sum_{x \in V_{n}^{(p)}}\zeta_{p}^{f(x)-\langle ca, x\rangle_{n}}=W_{f}(ca),
	\end{equation*}
	which implies that $cS_{f}=S_{f}$ and $\epsilon_{f}(ca)=\epsilon_{f}(a), f^{*}(ca)=f^{*}(a)$. Then by \cite[Theorem 1]{CM2013A}, 
	for any $c \in \mathbb{F}_{p}^{*}$, $S_{cf}=cS_{f}=S_{f}$, and for any $a \in S_{f}$, $\epsilon_{cf}(a)=\epsilon_{f}(c^{-1}a)=\epsilon_{f}(a), (cf)^{*}(a)=cf^{*}(c^{-1}a)=cf^{*}(a)$. By Theorem \ref{Theorem 4}, $\{D_{f, i}, i \in \mathbb{F}_{p}\}$ is an $s$-plateaued partition.
\end{proof}

\begin{remark}\label{Remark 6}
Keep the same notations as in Theorem \ref{Theorem 4}. Suppose $\Gamma=\{A_{i}, i \in V_{m}^{(p)}\}$ is an $s$-plateaued partition with $-A_{i}=A_{i}, i\in V_{m}^{(p)}$. Notice that for any $c \in {V_{m}^{(p)}}^{*}$, $\{D_{F_{c},j}, j \in \mathbb{F}_{p}\}$ is an $s$-plateaued partition. By Corollary \ref{Corollary 1}, $F_{c}, c \in {V_{m}^{(p)}}^{*}$, are all of $(p-1)$-form. Then we have $F(ax)=F(x)$ for all $a \in \mathbb{F}_{p}^{*}, x \in V_{n}^{(p)}$. Therefore, $\mathbb{F}_{p}^{*}A_{i}=A_{i}, i \in V_{m}^{(p)}$. Suppose $0 \in A_{i_{0}}$. Then $(p-1) \mid |A_{i}| \ (i \neq i_{0})$. By the value distribution of unbalanced plateaued functions, $\epsilon(0)=1$ if $p>3$, that is, $|A_{i_{0}}|=p^{n-m}-p^{\frac{n+s}{2}-m}+p^{\frac{n+s}{2}}, |A_{i}|=p^{n-m}-p^{\frac{n+s}{2}-m} \ (i \neq i_{0})$ if $p>3$.
\end{remark}

When $p$ is odd, $K \geq 5$ and $s=0$, we obtain the following characterization of bent partitions $\Gamma=\{A_{i}, 1 \leq i \leq K\}$ with $A_{i}=-A_{i}$. 

\begin{corollary}\label{Corollary 2}
	Let $p$ be an odd prime and $K \geq 5$ be an integer divisible by $p$. Let $\Gamma=\{A_{i}, 1 \leq i \leq K\}$ be a partition of $V_{n}^{(p)}$, where $-A_{i}=A_{i}, 1 \leq i \leq K$. Then the following statements are equivalent.
	
	(i) $\Gamma$ is a bent partition.
	
	(ii) $K$ is a power of $p$ (denoted by $K=p^m$), and if $\Gamma$ is written as $\Gamma=\{A_{i}, i \in V_{m}^{(p)}\}$, then $F: V_{n}^{(p)} \rightarrow V_{m}^{(p)}$ defined as 
	\begin{equation*}
		F(x)=i, x \in A_{i}, i \in V_{m}^{(p)}
	\end{equation*}
	is a vectorial bent function satisfying the following conditions:
	\begin{itemize}
		\item[(1)] Each component function $F_{c}$ ($c \in {V_{m}^{(p)}}^{*}$) satisfies $\epsilon_{F_{c}}(a)=\epsilon(a)$ for any $a \in V_{n}^{(p)}$, where $\epsilon(a)$ is independent of $c$. 
		\item[(2)] There is a function $G: V_{n}^{(p)} \rightarrow V_{m}^{(p)}$ such that $(F_{c})^{*}(x)=G_{c}(x), x \in V_{n}^{(p)}, c \in {V_{m}^{(p)}}^{*}$.
	\end{itemize}
\end{corollary}

\begin{proof}
For a $p$-ary bent function $f: V_{n}^{(p)} \rightarrow \mathbb{F}_{p}$, $S_{f}=V_{n}^{(p)}$. Then the result follows from Theorem \ref{Theorem 4}.
\end{proof}

 A \emph{partial spread} $\mathcal{S}$ of $V_{n}^{(p)}$ ($n=2m$) is a set of subspaces of $V_{n}^{(p)}$ of dimension $m$ for which the intersection of any two distinct subspaces is the zero subspace. If $|\mathcal{S}|=p^{m}+1$, then $\mathcal{S}$ is called a \emph{spread}. For the known bent partitions of $V_{n}^{(p)}$, the largest depth is $p^{\frac{n}{2}}$, which is obtained from spreads. In details, if $\{U_{0}, U_{1}, \dots, U_{p^m}\}$ is a spread of $V_{n}^{(p)} \ (n=2m)$, then $\{U_{0} \cup U_{1}, U_{j}^{*}, 2\leq j \leq p^m\}$ is a bent partition of $V_{n}^{(p)}$. There is an unsolved problem proposed in \cite{AM2022Be}:

\textbf{Open Problem}: If $\Gamma=\{A_{i}, i \in V_{n/2}^{(p)}\}$ is a bent partition of $V_{n}^{(p)}$ of depth $p^{\frac{n}{2}}$, must $\Gamma$ be obtained from spreads?

The following theorem partially addresses the above open problem.

\begin{theorem}\label{Theorem 5}
Let $p$ be an odd prime, $n \geq 4$ be an even positive integer. Let $\Gamma=\{A_{i}, i \in V_{n/2}^{(p)}\}$ be a bent partition of $V_{n}^{(p)}$ satisfying $-A_{i}=A_{i}$ for any $i$, and in the special case of $p=3$, $A_{i} \neq \{0\}$ for all $i$. Then the bent partition $\Gamma$ must be obtained from spreads.
\end{theorem}

\begin{proof}
Without loss of generality, we assume that $0 \in A_{0}$. Denote $U_{i}=A_{i} \cup \{0\}, i \in {V_{n/2}^{(p)}}^{*}$. Define $F: V_{n}^{(p)} \rightarrow V_{n/2}^{(p)}$ as $F(x)=i, x \in A_{i}, i \in V_{n/2}^{(p)}$. By Corollary \ref{Corollary 2} and the proof of Theorem \ref{Theorem 4}, $(F_{c})^{*}(x)=G_{c}(x), \epsilon_{F_{c}}(x)=\epsilon(x), x \in V_{n}^{(p)}, c \in {V_{n/2}^{(p)}}^{*}$, and for any $a \in V_{n}^{(p)}$ and $i \in V_{n/2}^{(p)}$, 
\begin{equation*}
	\chi_{-a}(A_{i})=p^{\frac{n}{2}}\delta_{0}(a)-\epsilon(a)+\delta_{i}(G(a))\epsilon(a)p^{\frac{n}{2}},
\end{equation*}
where $G$ is some function from $V_{n}^{(p)}$ to $V_{n/2}^{(p)}$ and $\epsilon(a) \in \{\pm 1\}$ is independent of $c$. By Remark \ref{Remark 6} and the condition that $A_{i} \neq \{0\}$ for all $i \in V_{n/2}^{(p)}$ when $p=3$, we have $\epsilon(0)=1$ and $|A_{i}|=p^{\frac{n}{2}}-1, i \in {V_{n/2}^{(p)}}^{*}$. We claim that $\epsilon(a)=1$ for all $a \in V_{n}^{(p)}$. To the contrary, assume that there exists $a \in {V_{n}^{(p)}}^{*}$ such that $\epsilon(a)=-1$. Let $i \in V_{n/2}^{(p)}$ be fixed with $i \neq 0, G(a)$. Then 
\begin{equation*}
		1=\chi_{-a}(A_{i})=\sum_{x \in A_{i}}\zeta_{p}^{\langle -a, x\rangle_{n}}=\sum_{j \in \mathbb{F}_{p}}n_{j}\zeta_{p}^{j},
\end{equation*}
where $n_{j}=|\{x \in A_{i}: \langle -a, x\rangle_{n}=j\}|$. Since $1+x+x^{2}+\dots+x^{p-1}$ is the minimal polynomial of $\zeta_{p}$ over $\mathbb{Q}$, then $n_{0}-1=n_{1}=n_{2}=\dots=n_{p-1}$. Then by $|A_{i}|=p^{\frac{n}{2}}-1$,  $n_{1}=\frac{p^{\frac{n}{2}}-2}{p}$ is an integer, which is impossible. Thus, $\epsilon(a)=1$ for any $a \in V_{n}^{(p)}$. Then for any $a \in V_{n}^{(p)}$ and $i \in {V_{n/2}^{(p)}}^{*}$, $\chi_{-a}(A_{i}) \in \{p^{\frac{n}{2}}-1, -1\}$ and thus $\chi_{-a}(U_{i}) \in \{p^{\frac{n}{2}}, 0\}$. Let $i \in {V_{n/2}^{(p)}}^{*}$ be fixed. From $|U_{i}|=p^{\frac{n}{2}}$ and $0 \in U_{i}$, $\chi_{-a}(U_{i})=\sum_{x \in U_{i}}\zeta_{p}^{-\langle a, x\rangle_{n}}=p^{\frac{n}{2}}$ if and only if $a \in U_{i}^{\perp}$, where $U_{i}^{\perp}=\{a \in V_{n}^{(p)}: \langle a, x\rangle_{n}=0 \text{ for all } x \in U_{i}\}$ is a linear subspace. Denote $m=|\{a \in V_{n}^{(p)}: \chi_{-a}(U_{i})=p^{\frac{n}{2}}\}|$. Then $p^{n}=\sum_{x \in U_{i}}\sum_{a \in V_{n}^{(p)}}\zeta_{p}^{-\langle a, x\rangle_{n}}=\sum_{a \in V_{n}^{(p)}}\chi_{-a}(U_{i})=p^{\frac{n}{2}}m$, which implies that $m=p^{\frac{n}{2}}$. Thus, the dimension of $U_{i}^{\perp}$ is $\frac{n}{2}$. Note that $<U_{i}>^{\perp}=U_{i}^{\perp}$, where $<U_{i}>$ is the linear space spanned by $U_{i}$. Then the dimension of $<U_{i}>$ is $\frac{n}{2}$, and by $|U_{i}|=p^{\frac{n}{2}}$, we have that $U_{i}$ is a linear subspace of dimension $\frac{n}{2}$. Obviously, $U_{i} \cap U_{j}=\{0\}$ for any $i \neq j \in {V_{n/2}^{(p)}}^{*}$. Thus, $\{U_{i}, i \in {V_{n/2}^{(p)}}^{*}\}$ is a partial spread. 
By \cite[Corollary 3.8]{Metsch1991Im}, the partial spread $\{U_{i}, i \in {V_{n/2}^{(p)}}^{*}\}$ can be extended to a spread $\{U_{i}, i \in {V_{n/2}^{(p)}}^{*}, W, V\}$. Then $A_{0}=W\cup V$. Therefore, $\{A_{i}, i \in V_{n/2}^{(p)}\}=\{W \cup V, U_{i}^{*}, {i \in V_{n/2}^{(p)}}^{*}\}$, which is obtained from the spread $\{U_{i}, i \in {V_{n/2}^{(p)}}^{*}, W, V\}$. This completes the proof.
\end{proof}

\section{Conclusion}
Plateaued functions are an important class of cryptographic functions. In this paper, we introduced $s$-plateaued partitions $\Gamma=\{A_{i}, 1 \leq i \leq K\}$ of $V_{n}^{(p)}$ of depth $K$, which play a powerful role in constructing $p$-ary $s$-plateaued functions, vectorial $s$-plateaued functions and generalized $s$-plateaued functions. We analyzed the possible cardinality of $A_{i}$, which is far more complex than bent partitions. Based on a general way proposed, we presented some explicit constructions of $s$-plateaued partitions for which any generated $p$-ary $s$-plateaued function has no nonzero linear structure. When $p$ is odd, $K \geq 5$ and $A_{i}=-A_{i}, 1 \leq i \leq K$, we characterized the corresponding $s$-plateaued partitions. We further showed that when $p \geq 5$, the preimage set partition of a $p$-ary $s$-plateaued function $f: V_{n}^{(p)} \rightarrow \mathbb{F}_{p}$ with $f(x)=f(-x)$ is an $s$-plateaued partition if and only if $f$ is of $(p-1)$-form, where $n+s$ is even. For the open problem on the relations between bent partitions of depth $p^{\frac{n}{2}}$ and spreads, we partially provided an answer. 

We list some questions for further research:
\begin{itemize}
\item The constructed $s$-plateaued partitions in this paper are all unbalanced. Is there an $s$-plateaued partition of depth $K>4$ that generates balanced $p$-ary $s$-plateaued functions?
\item When $p$ is odd, we have proven that all $p$-ary $s$-plateaued functions derived from an $s$-plateaued partition with $A_{i}=-A_{i}$ have the same Walsh support. It remains an open problem whether this property extends to other $s$-plateaued partitions.
\item Bent partitions have rich connections with coding theory and combinatorics. For $s$-plateaued partitions ($s>0$), are there some profound connections with coding theory and combinatorics?
\end{itemize}


\begin{thebibliography}{}
\bibitem{AAKM2024Be}
S. Alkan, N. Anbar, T. Kalayci, W. Meidl, Bent partitions and LP-packings, IEEE Trans. Inf. Theory, vol.70, no.7, pp.5365-5375, 2024.
\bibitem{AAKM2025Be}
S. Alkan, N. Anbar, T. Kalayci, W. Meidl, Bent Partition, vectorial dual-bent function, and LP-packing constructions, IEEE Trans. Inf. Theory, vol.71, no.1, pp.752-767, 2025.
\bibitem{AFKMWW2026A}
N. Anbar, F.-W. Fu, T. Kalayci, W. Meidl, J. Wang, Y. Wei, Analysis of some classes of bent partitions and vectorial bent functions, Des. Codes Cryptogr, vol.94, no.4: 78, 2026.
\bibitem{AKM2024Am}
N. Anbar, T. Kalayci, W. Meidl, Amorphic association schemes from bent partitions, Discret. Math., 
vol.347, no.1: 113658, 2024.
\bibitem{AKM2022Be}
N. Anbar, T. Kalayci, W. Meidl, Bent partitions and partial difference sets, IEEE Trans. Inf. Theory, 
vol.68, no.10, pp.6894-6903, 2022.
\bibitem{AKM2023Ge}
N. Anbar, T. Kalayci, W. Meidl, Generalized semifield spreads, Des. Codes Cryptogr., vol.91, no.2, 
pp.545-562, 2023.
\bibitem{AKMO2025Be}
N. Anbar, T. Kalayci, W. Meidl, F. \"{O}zbudak, Bent partitions and Maiorana-McFarland association schemes, Cryptogr. Commun., vol.17, no.6, pp.1641-1657, 2025.
\bibitem{AM2022Be}
N. Anbar, W. Meidl, Bent partitions, Des. Codes Cryptogr., vol.90, no.4, pp.1081-1101, 2022.
\bibitem{BFT2018Ex}
S. Bozta\c{s}, F. \"{O}zbudak, E. Tekin, Explicit full correlation distribution of sequence families using plateaued functions, IEEE Trans. Inf. Theory, 
vol.64, no.4, pp.2858-2875, 2018.
\bibitem{Carlet2015Bo}
C. Carlet, Boolean and vectorial plateaued functions, and APN functions. IEEE Trans. Inf. Theory, vol.61, no.11, pp.6272-6289, 2015.
\bibitem{Carlet1993Pa}
C. Carlet, Partially-bent functions, Des. Codes Cryptogr., vol.3, no.2, pp.135-145, 1993.
\bibitem{CM2013A}
A. \c{C}e\c{s}melio\u{g}lu, W. Meidl, A construction of bent functions from plateaued functions, Des. Codes Cryptogr., vol.66, nos.1-3,
pp.231-242, 2013.
\bibitem{CO2021Gr}
A. \c{C}e\c{s}melio\u{g}lu, O.\"{O}lmez, Graphs of vectorial plateaued functions as difference sets, Finite Fields Their Appl.,
vol.71:101795, 2021.
\bibitem{Feng2018Al}
K. Feng, Algebraic Number Theory. Harbin, China: Harbin Institute
Technology Press, 2018.
\bibitem{HPWZ2019De}
S. Hod\v{z}i\'{c}, E. Pasalic, Y. Wei, F. Zhang, Designing plateaued Boolean functions in spectral domain and their classification, IEEE Trans. Inf. Theory, vol.65, no.9, pp.5865-5879, 2019.
\bibitem{IR1990A}
K. Ireland, M. Rosen, A Classical Introduction to Modern Number
Theory. New York, NY, USA: Springer, 1990.
\bibitem{Mesnager2016Be}
S. Mesnager, Bent Functions-Fundamentals and Results, Springer, Switzerland, 2016.
\bibitem{MTQ2017Ge}
S. Mesnager, C. Tang, Y. Qi,
Generalized plateaued functions and admissible (plateaued) functions, IEEE Trans. Inf. Theory, vol.63, no.10, pp.6139-6148, 2017.
\bibitem{MOS2019Li}
S. Mesnager, F. \"{O}zbudak, A. S{\i}nak, Linear codes from weakly regular plateaued functions and their secret sharing schemes, Des. Codes Cryptogr., vol.87, nos.2-3, pp.463-480, 2019.
\bibitem{MOS2016Re}
S. Mesnager, F. \"{O}zbudak, A. S{\i}nak, Results on characterizations of plateaued functions in arbitrary characteristic. In Proceedings of BalkanCryptSec 2015, Lecture Notes in Computer Science, vol.9540, pp.17-30, 2016.
\bibitem{MS2020Se}
S. Mesnager, A. S{\i}nak, Several classes of minimal linear codes with few weights from weakly regular plateaued functions, IEEE Trans. Inf. Theory, vol.66, no.4, pp.2296-2310, 2020.
\bibitem{Metsch1991Im}
K. Metsch, Improvement of Bruck's completion theorem, Des. Codes Cryptogr., vol.1, no.2, pp.99-116, 1991.
\bibitem{Olmez2015Pl}
O. \"{O}lmez, Plateaued functions and one-and-half difference sets, Des. Codes Cryptogr., vol.76, no.3, pp.537-549, 2015.
\bibitem{OP2020Du}
F. \"{O}zbudak, R. M. Pelen, Duals of non-weakly regular bent functions
are not weakly regular and generalization to plateaued functions, Finite Fields Their Appl., vol.64: 101668, 2020. 
\bibitem{RPZW2024Se}
R. Rodr\'{i}guez-Aldama, E. Pasalic, F. Zhang, Y. Wei, Self-orthogonal minimal codes from (vectorial) $p$-ary plateaued functions, IACR Cryptol. ePrint Arch. 2024: 1383.
\bibitem{WF2023Ne}
J. Wang, F.-W. Fu, New results on vectorial dual-bent functions and partial difference sets, Des. Codes Cryptogr., vol.91, no.1, pp.127-149, 2023.
\bibitem{WFW2023Be}
J. Wang, F.-W. Fu, Y. Wei, Bent Partitions, vectorial dual-bent functions and partial difference sets, IEEE Trans. Inf. Theory, vol.69, no.11, pp.7414-7425, 2023.
\bibitem{WWF2026Fu}
J. Wang, Y. Wei, F.-W. Fu, Further results on bent partitions, IEEE Trans. Inf. Theory, vol.72, no.5, pp.3447-3459, 2026.
\bibitem{WFWY2024A}
J. Wang, F.-W. Fu, Y. Wei, J. Yang, A further study of vectorial dual-bent functions, IEEE Trans. Inf. Theory, vol.70, no.10, pp.7472-7483, 2024.
\bibitem{WWF2026Se}
Y. Wei, J. Wang, F.-W. Fu, Self-orthogonal codes from plateaued functions and their applications in quantum codes and LCD codes, IEEE Trans. Inf. Theory, vol.72, no.3, pp.1629-1645, 2026.
\bibitem{ZPZ2022Ph}
W. Zhang, E. Pasalic, L. Zhang, Phase orthogonal sequence sets for (QS)CDMA communications, Des. Codes Cryptogr., vol.90, no.5, pp.1139-1156, 2022.
\bibitem{ZZ2001On}
Y. Zheng, X. M. Zhang, On plateaued functions, IEEE Trans. Inf. Theory, vol.47, no.3, pp.1215-1223, 2001.

\end{thebibliography}
\end{document}